\documentclass[runningheads]{llncs}
 
\usepackage[utf8]{inputenc}
\usepackage{graphicx,color}
\usepackage{amsmath,amssymb}
\usepackage{mathrsfs}
\usepackage[T1]{fontenc}
\usepackage{tikz}
\usepackage{pict2e}

\usetikzlibrary{decorations,decorations.pathmorphing,arrows,arrows.meta,shapes,automata,backgrounds,fit,calc,petri,patterns,matrix,positioning}

\pgfarrowsdeclare{arrow30}{arrow30}
{
  \arrowsize=0.3pt
  \advance\arrowsize by.275\pgflinewidth%
  \pgfarrowsleftextend{+-\arrowsize}
  \advance\arrowsize by.5\pgflinewidth
  \pgfarrowsrightextend{+\arrowsize}
}
{
  \arrowsize=0.3pt
  \advance\arrowsize by.275\pgflinewidth%
  \pgfsetdash{}{+0pt}
  \pgfsetroundjoin
  \pgfpathmoveto{\pgfqpoint{1\arrowsize}{0\arrowsize}}
  \pgfpathlineto{\pgfqpoint{-12\arrowsize}{2.5\arrowsize}}
  \pgfpathlineto{\pgfqpoint{-12\arrowsize}{-2.5\arrowsize}}
  \pgfpathclose
  \pgfusepathqfillstroke
}

\pgfdeclarearrow{
  name=read,
  parameters= {\the\pgfarrowlength},  
  setup code={
   \pgfarrowssettipend{0pt}
   \pgfarrowssetlineend{-\pgfarrowlength}
   \pgfarrowlinewidth=\pgflinewidth
   \pgfarrowssavethe\pgfarrowlength
  },
  drawing code={
   \pgfpathcircle{\pgfpoint{-0.5\pgfarrowlength}{0pt}}{0.5\pgfarrowlength}
   \pgfpathmoveto{\pgfpoint{-0.85355\pgfarrowlength}{0.35355\pgfarrowlength}}
   \pgfpathlineto{\pgfpoint{-0.14644\pgfarrowlength}{-0.35355\pgfarrowlength}}
   \pgfpathmoveto{\pgfpoint{-0.14644\pgfarrowlength}{0.35355\pgfarrowlength}}
   \pgfpathlineto{\pgfpoint{-0.85355\pgfarrowlength}{-0.35355\pgfarrowlength}}
   \pgfusepathqstroke
  },
  defaults = { length = 7pt }
}
\tikzset{
    read/.style={
      draw,circle,fill=white,inner sep=0pt,minimum size=7pt,outer sep=0pt,
      path picture={
         \draw (path picture bounding box.north west) -- (path picture bounding box.south east);
         \draw (path picture bounding box.north east) -- (path picture bounding box.south west);
            }
        },
}

\newcommand{\pre}[1]{{^\bullet} {#1}}
\newcommand{\post}[1]{{#1} {^\bullet}}

\def\Ag{A}
\def\PR{\mathit{PR}}  
\def\TR{\mathit{TR}}  
\def\TS{\mathit{TS}}
\def\Trext{\mathit{Tr}} 
\def\Evt{\mathit{Evt}} 
\def\Lbl{\mathit{Evt}} 
\def\Ls{\mathit{L}} 
\def\MG{\mathit{MG}} 

\def\read{\mathsf{read}}
\def\inh{\mathsf{inh}}
\def\flow{\mathsf{flow}}
\def\AMAS{\mathrm{AMAS}}

\newsavebox{\ORCIDlogo}
\savebox{\ORCIDlogo}{%
\setlength{\unitlength}{\dimexpr 1em/256\relax}%
\begin{picture}(256,256)%
  \color[HTML]{A6CE39}\put(128,128){\circle*{256}}%
  \color{white}%
  \put(78.6,199.2){\circle*{20}}%
  \moveto(70.9,176,9)\lineto(86.3,176,9)\lineto(86.3,69.8)\lineto(70.9,69.8)%
  \closepath\fillpath%
  \moveto(108.9,176.9)\lineto(150.5,176.9)%
  \curveto(190.1,176.9)(207.5,148.6)(207.5 ,123.3)%
  \curveto(207.5,95,8)(186,69.7)(150.7,69.7)%
  \lineto(108.9,69.7)%
  \closepath\fillpath%
  \color[HTML]{A6CE39}%
  \moveto(124.3,83.6)\lineto(148.8,83.6)%
  \curveto(183.7,83.6)(191.7,110.1)(191.7,123.3)%
  \curveto(191.7,144.8)(178,163)(148,163)%
  \lineto(124.3,163)%
  \closepath\fillpath%
\end{picture}%
}

\newcommand\orcidicon[1]{\href{https://orcid.org/#1}{\usebox{\ORCIDlogo}}}
\usepackage{hyperref}

\begin{document}
\title{Asynchronous Multi-Agent Systems with Petri nets}
\author{Federica Adobbati\inst{2}\orcidicon{0000-0002-6356-7026} \and
{\L}ukasz Mikulski\inst{1}\orcidicon{0000-0002-6711-557X} 
}
\authorrunning{F. Adobbati and {\L}. Mikulski} 
\institute{
Faculty of Mathematics and Computer Science, 
Nicolaus Copernicus University, Chopina 12/18, Toru{\'n}, Poland\\
\and
National Institute of Oceanography and Applied Geophysics - OGS, Trieste, Italy \\
\email{fadobbati@ogs.it}\\
\email{lukasz.mikulski@mat.umk.pl}\\
}
\maketitle 
\begin{abstract}
Modeling the interaction between components is crucial for many applications and serves as a fundamental step in analyzing and verifying properties in multi-agent systems.

In this paper, we propose a method based on 1-safe Petri nets to model Asynchronous Multi-Agent Systems (AMAS), starting from two semantics defined on AMAS represented as transition systems. Specifically, we focus on two types of synchronization: synchronization on transitions and synchronization on data. For both, we define an operator that composes 1-safe Petri nets and demonstrate the relationships between the composed Petri net and the global transition systems as defined in the literature.

Additionally, we analyze the relationships between the two semantics on Petri nets, proposing two constructions that enable switching between them. These transformations are particularly useful for system analysis, as they allow the selection of the most suitable model based on the property that needs to be verified.

\keywords{multi-agent systems, Petri nets, 1-safe, synchronization, asynchronous composition} 
\end{abstract}
\section{Introduction}
Multi-agent systems (MAS) allow one to model and analyze the interactions between several components.
Their formal study is relevant for several problems, such as, for example, the work on games, where the goal is to understand whether a set of agents is able to enforce certain properties on a global system on which they interact (e.g. \cite{AHK02,lomuscio2017mcmas,belardinelli2023abstraction}).
When modeling MAS, several studies (e.g. \cite{KJMMPPS23,JP10,W09,WH03} 
) consider synchronous systems, that is, they assume the existence of a global clock. At every instant, each agent selects its own action, and when the clock ticks, all the actions are executed together, determining the next global state of the system.
In contrast, in an asynchronous system, agents can perform their actions whenever they become available, without following any predefined order. 
Asynchronous systems are convenient in many applications; however, their analysis presents some challenges.
Although each agent may have a limited number of states, the global system obtained composing the agents together can be exponentially larger, especially in systems with a high level of concurrency. This is known as \emph{state explosion problem}, and several works have been developed to address it.
In \cite{alberdingk2023modular}, the authors develop a modular control plane verification, where they decompose the global system into smaller components and verify properties on them; in \cite{mikulski2022assume} the authors propose an assume-guarantee scheme to verify strategic properties. 

Petri nets are formal models that can explicitly represent concurrency. They were introduced in the 60 by Carl Adam Petri \cite{Petri62}, and their use allows to efficiently represent distributed systems, and to study them with concurrent semantics rather than interleaving.
Several classes of Petri nets have been developed in recent decades, allowing us to model different features of the systems.
In this work, we focus on 1-safe nets with the addition of read arcs \cite{montanari1995}.

There is a strong relationship between Petri nets and transition systems. Given a Petri net, obtaining a transition system describing its behavior is always easy: each state of the transition system is a global state of the Petri net, and states are connected by arcs labeled as the transitions in the Petri net producing the change in the global state. 
The reverse process namely, given a transition system how to construct a Petri net with the same behavior, is known as the synthesis problem, and it is in general not always possible.
In \cite{BBD15} the authors describe how to obtain a Petri net from a transition system for different classes of Petri nets.
Multi-Agent systems can be described as Petri nets, and their composition is generally smaller than in the case of transition systems, due to the explicit representation of concurrency in Petri nets.
The composition of Petri nets can be executed in several different ways, based on the kinds of interactions between components and properties that the composition needs to satisfy. Several works in the literature address the problem of composition: in \cite{bernardinello2023,nesterov2019}, the authors present a composition that preserves soundness between workflow nets; \cite{baldan2001} presents a composition operation between open nets, i.e. Petri nets with a set of places that represent the interface of the net with the external environment; \cite{haddad2013} introduces I/O Petri nets studying the properties of their channels. 
In \cite{reisig2019},
a general framework to model composition is proposed, and,
among other formalisms, its application to Petri nets is shown. 
In \cite{fettke2022} the authors present an application of composition calculus to business process models.

In our work, we propose a framework to study MAS based on Petri nets. We start from two semantics developed on transition systems, and we show how we can derive Petri net agents and their composition starting from each of them.
We then prove that from the Petri net MAS in one of these semantics, we can always obtain an equivalent Petri net MAS in the other. The steps to reach this goal are described in Fig.~\ref{fig:schema}. We start from AMAS and open MAS.
For both semantics, the transition system of each agent can be synthesized into a 1-safe Petri net. In AMAS, agents synchronize on common transitions, and the resulting global model is defined as canonical interleaved interpreted systems (IIS). In \cite{AM23} we proved that composing Petri nets agents through fusion of transitions produces a Petri net with a marking graph isomorphic to the canonical IIS.
When we consider open MAS, the translation into Petri nets require more steps. In an open MAS, the possibility to execute an action is conditioned on the value of some external variables. This concept is not explicitly present in Petri net semantics. For this reason, at agent level, we prove the equivalence between closed module, namely modules in which external dependencies have been explicitly added, and Petri net agents modified in order to explicitly add constraints derived from compositions.
In addition, we prove the equivalence between the global closed MAS obtained through data synchronization and the global Petri net obtained by fusion of sequential components.
Finally, we prove a relation between the two semantics.

The paper is structured as follows. In Sec.~\ref{sec:mas-ts} we formally define AMAS, open MAS, closed modules and their composition.
In Sec.~\ref{sec:pn} we provide basic definitions on 1-safe Petri nets and their synthesis from transition systems.
Sec.~\ref{sec:comp-pn} presents our first contribution: we describe agent composition for the two models, proving the equivalence between the Petri net global models, and the canonical IIS or closed MAS.
In Sec.~\ref{sec:2sem-pn} we discuss and prove the relations between the two semantics; Sec.~\ref{sec:concl} concludes the paper describing future developments of this work. 
\begin{figure}[ht]
    \centering
  \begin{tikzpicture}[node distance=2cm,>=arrow30]
    \node(x1){AMAS}; 
\node(x2)[below of=x1]{open MAS (module)}; 
\node(x3)[right of=x2, yshift=1cm]{Tran. sys.};
\node(x4)[right of=x3, xshift=1cm]{1-safe Petri net};
\node(x5)[right of=x4, xshift=1cm, yshift=-1cm]{net with var.};
\node(x6)[below of=x5, yshift=.5cm]{net with int.};
\node(y2)[above of=x1, xshift=1cm]{canonical IIS};
\node(y3)[below of=x2, xshift=1cm, yshift=-2.5cm]{closed MAS};
\node(y4)[right of=y2, xshift=2cm]{global net (dyn. synch.)};
\node(y5)[right of=y3, xshift=2cm]{global net (stat. synch.)};
\node(y6)[below of=x2, xshift=2.0cm, yshift=0cm]{closed module};
\node(y7)[right of=y6, xshift=1cm, yshift=0cm]{net with ext.};
\node(z1)[right of=y2, xshift=-.5cm]{$\equiv$};
\node(z2)[right of=y3, xshift=-.5cm]{$\equiv$};
\node(z3)[right of=y6, xshift=-.5cm]{$\equiv$};
\draw(x1) edge [->] node (t1) [xshift=.5cm, yshift=.2cm] {\tiny{proj.}} (x3);
\draw(x2) edge [->] node (t2) [xshift=.7cm, yshift=-.2cm] { } (x3);
\draw(x3) edge [->] node (t3) [yshift=.3cm] {\tiny{synthesis}} (x4);
\draw(x4) edge [->, out=-60, in=170] (x5);
\draw(x5) edge [->,loop right] node (t4) {\tiny{refine}} (x5);
\draw(x6) edge [->,loop right] node (t5) {\tiny{refine}} (x6);
\draw(x1) edge [->] node (t6) [xshift=1.3cm] {\tiny{synch. on transitions}} (y2);
\draw(x2) edge [->] node (t7) [xshift=-1cm] {\tiny{synch. on data}} (y3);
\draw(x4) edge [->] node (t8) [xshift=.8cm, yshift=.2cm] {\tiny{tran. fusion}} (y4);
\draw(x6) edge [->] node (t9) [xshift=1.1cm, yshift=.1cm] {\tiny{place fusion}} (y5);
\draw(x2) edge [->, out=0, in=170] node (t10) [yshift=.2cm, xshift=.1cm] {\tiny{add seq. comp.}} (x5);
\draw(x5) edge [->, out=-120, in=120] node (t11) [xshift=-1.2cm] {\tiny{add ext. places}} (x6);
\draw(x2) edge [->, out=0, in=120] (x6);
\draw(x2) edge [->] node (t12) [xshift=1.0cm, yshift=.1cm] {\tiny{fair randomness}} (y6);
\draw(x6) edge [->] node (t13) [xshift=-0.7cm, yshift=.3cm] {\tiny{fair randomness}} (y7);
\draw(y7) edge [->] node (t14) [xshift=-1.1cm, yshift=.1cm] {\tiny{seq. comp. fusion}} (y5);
\end{tikzpicture}
    \caption{Schema of involved models.}
    \label{fig:schema}
\end{figure}
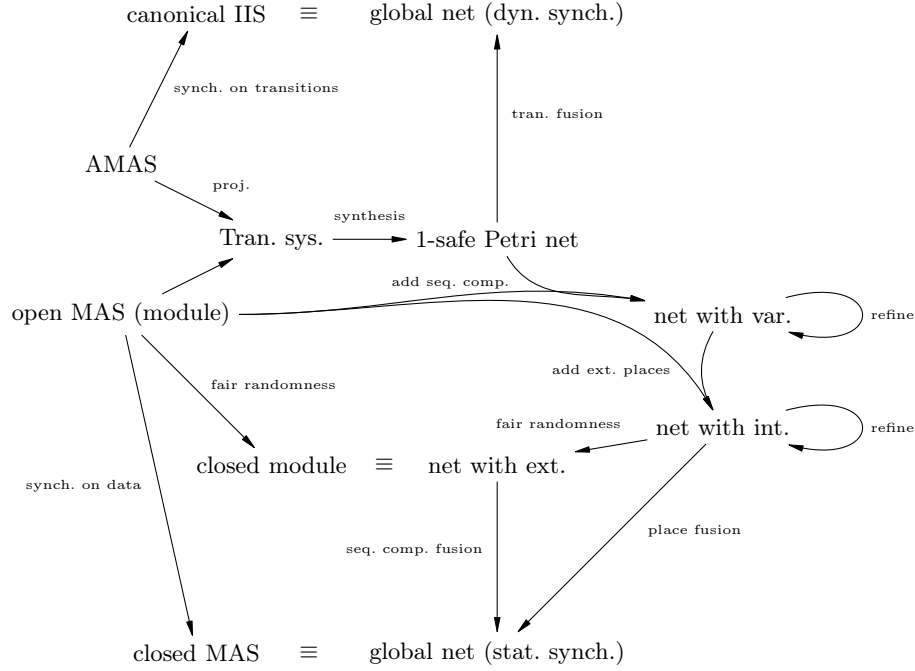

\section{MAS as transition systems: two semantics}
\label{sec:mas-ts}
Transition systems are often used in the formal analysis of multi-agent systems.
In these models, each agent is represented as a transition system, and the
rules to compose them to get the global model may vary based on what 
the model intends to explicitly represent.
In this paper we focus on two different composition rules, referring to two
different semantics, both aiming to describe asynchronous MAS. 
In the first, described for example in \cite{JPSDM20}, 
the synchronization is based
on common transitions; in the second, described in \cite{LSWW13},
the synchronization is based on the read only access to the local states of other agents.
These two semantics and corresponding compositions
are described more in detail in the following subsections.

\begin{example}
Fig.~\ref{fig:tramas} shows an example of multi-agent system.
In this model, the agents on the left and on the right represent
two trains, whereas the agent in the middle is the controller.
Each train has three local states, for each $i\in\{1, 2\}$, $w_i$ is the state in which the train $i$ is waiting to enter
a tunnel, in $t_i$, train $i$ is in the tunnel, and in $a_i$
train $i$ is out of the tunnel.
The controller has three states: $g$ denotes that no train is in
the tunnel (green light for entering), $r_1$ denote the presence of train $1$ in the tunnel
or its permission to enter it (red light for others), and analogously for $r_2$ with
respect to train $2$.
The synchronizations between agents need to guarantee that the 
two trains are never in the tunnel together.
The example can be easily generalized for an arbitrary number 
of trains.
It is discussed both in \cite{JPSDM20} and in \cite{MJK23}, where two
different semantics of synchronization are proposed, leading
to different behaviors of the global system.
\end{example}
\begin{figure}
    \centering
  \begin{tikzpicture}[->,auto,>=arrow30,node distance=8mm,font=\tiny]

\node[circle,draw=black,inner sep=1pt,minimum size=12pt] (A1) {w1};
\node[circle,draw=black,inner sep=1pt,minimum size=12pt,below of=A1, xshift=-8mm] (A2) {a1};
\node[circle,draw=black,inner sep=1pt,minimum size=12pt,below of=A2, xshift=8mm] (A3) {t1};

\node[circle,draw=black,inner sep=1pt,minimum size=12pt,right of=A1, xshift=12mm] (B1) {g};
\node[below of=B1](BT){};
\node[circle,draw=black,inner sep=1pt,minimum size=12pt,below of=BT, xshift=-8mm] (B2) {r1};
\node[circle,draw=black,inner sep=1pt,minimum size=12pt,below of=BT, xshift=8mm] (B3) {r2};

\node[circle,draw=black,inner sep=1pt,minimum size=12pt,right of=B1, xshift=12mm] (C1) {w2};
\node[circle,draw=black,inner sep=1pt,minimum size=12pt,below of=C1, xshift=8mm] (C2) {a2};
\node[circle,draw=black,inner sep=1pt,minimum size=12pt,below of=C2, xshift=-8mm] (C3) {t2};

\path (A1) edge node {n1} (A3)
      (A3) edge node {n2} (A2)
      (A2) edge node {n3} (A1)
      (B1) edge [bend left, out=15, in=165] node [xshift=-3mm, yshift=-2mm] {n1} (B2)
      (B2) edge [bend left, out=15, in=165] node [xshift=3mm, yshift=2mm] {n2} (B1)
      (B1) edge [bend left, out=15, in=165] node [xshift=-3mm, yshift=2mm] {m1} (B3)
      (B3) edge [bend left, out=15, in=165] node [xshift=3mm, yshift=-2mm] {m2} (B1)
      (C1) edge node [swap] {m1} (C3)
      (C3) edge node [swap] {m2} (C2)
      (C2) edge node [swap] {m3} (C1)
      ;

\end{tikzpicture}
    \caption{Train-Gate-Controller benchmark with two trains, adapted from \cite{JPSDM20}.}
    \label{fig:tramas}
\end{figure}
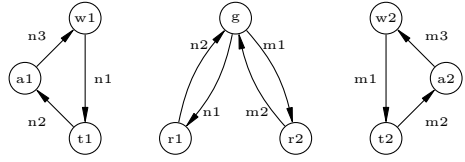
\subsection{Synchronization on transitions}
\label{sec:mas-tr}
In this section we recall an \emph{asynchronous multi-agent system} defined in~\cite{JPSDM20}.
\begin{definition}
An asynchronous multi-agent system ($\AMAS$) consists of n agents $\Ag =\{1,\dots,n \}$. 
Each agent is associated with a tuple 
$\Ag_i=(L_i,\iota_i,\Evt_i,\PR_i,\TR_i,\mathcal{PV}_i,V_i)$, where
\begin{itemize}
    \item $\Ls_i=\{l_i^1,\ldots,l_i^{n_i}\}$ is a set of local states;
    \item $\iota_i\in L_i$ is an initial state;
    \item $\Evt_i=\{\alpha_i^1,\ldots,\alpha_i^{m_i}\}$ is a set of events in which agent $\Ag_i$ can choose to participate;
    \item $\PR_i:L_i\to 2^{\Evt_i}$ is a local protocol, which assigns 
    events to states in which they are available;
    \item $\TR_i:L_i\times \Evt_i\to L_i$ is a local transition function, such that $\TR_i(l_i,\alpha)$ is 
    defined whenever $\alpha\in \PR_i(l_i)$;
    \item $\mathcal{PV}_i$ is a set of local propositions;
    \item $V_i:L_i\to 2^{\mathcal{PV}_i}$ is the valuation of local propositions in local states.
\end{itemize}
\end{definition} 
Since verification of the model of $\AMAS$ is not in the scope of this paper, we are interested only in
transition systems that define the behavior of $\AMAS$, namely the tuples $(\Ls_i,\Evt_i,\TR_i,\iota_i)$.
However, note that local events of different agents may be not disjoint. 
Events which are present in more than one event set $\Evt$ require participation of more than one 
agent, namely those agents synchronize on such events. 
\begin{example}
In Fig.~\ref{fig:tramas}, the events $n1, m1, n2, m2$ are shared
by two agents and require the participation of both of them to occur,
whereas the events $n3$ and $m3$ are local, and depends on a single agent. 
\end{example}

%
%
\paragraph{Composition of agents}
We recall from~\cite{JPSDM20} the definition of \emph{canonical interleaved interpreted system},
which is a composition of agents being parts of an asynchronous multi-agent system 
with synchronizations on common events.
Since we are not interested in model checking, we would concentrate only on the behavior
of the multi-agent system (putting apart the propositional variables). 
\begin{definition}
A canonical interleaved interpreted system (canonical IIS) is an $\AMAS$ extended with a tuple $(St,Evt,\TR,\iota)$ where:
\begin{itemize}
    \item $St\subseteq \Ls_1\times\ldots\times \Ls_n$ is a set of global states;
    \item $\Evt = \bigcup_{i\in\{1,\ldots,n\}} Evt_i$ is a set of events;
    \item $\TR:St\times Evt\to St$ is a (partial) global transition function, where
    $\TR((l_1,\ldots,l_n),\alpha)=(l'_1,\ldots,l'_n)$ if \ 
    $\TR_i(l_i,\alpha) = l'_i$ for all $i$ where $\alpha\in Evt_i$ and
    $\TR_i(l_i,\alpha) = l_i$ otherwise;
    \item $\iota=(\iota_1,\ldots,\iota_n)$ is an initial state.
\end{itemize} 
\end{definition}
Given a canonical IIS $I$, some of its states may not be reachable through any execution, 
due to the restrictions given by the synchronizations, 
and therefore also the transitions outgoing from these states can never be executed. 
We will denote with $I_r$ the canonical system where these unreachable states and transitions have been pruned. 

By definition, the number of states in the IIS grows exponentially with the number of agents, 
therefore limiting the number of compositions when studying the properties of the system may help in the analysis. 
\begin{example}
Fig.~\ref{fig:iis} shows the global model of the $\AMAS$ in 
Fig.~\ref{fig:tramas} when the synchronizations happen on transitions.
For clarity, in the figure we omitted the states unreachable 
from the initial one.
\end{example}
\begin{figure}
    \centering
  \begin{tikzpicture}[->,auto,>=arrow30,node distance=1.3cm,font=\tiny]

\node[initial,state,inner sep=-1pt,ellipse,minimum width=43pt] (A1)                {w1,g,w2};
\node[state,inner sep=-1pt,ellipse,minimum width=43pt]         (B2) [below of=A1, xshift=-2cm]  {t1,r1,w2};
\node[state,inner sep=-1pt,ellipse,minimum width=43pt]         (A2) [below of=B2]  {a1,g,w2};
\node[state,inner sep=-1pt,ellipse,minimum width=43pt]         (C2) [below of=A1,xshift=2cm]  {w1,r2,t2};
\node[state,inner sep=-1pt,ellipse,minimum width=43pt]         (A3) [below of=C2]  {w1,g,a2};
\node[state,inner sep=-1pt,ellipse,minimum width=43pt]         (D2) [below of=A3]  {t1,r1,a2};
\node[state,inner sep=-1pt,ellipse,minimum width=43pt]         (A4) [below of=D2,xshift=-2cm]  {a1,g,a2};
\node[state,inner sep=-1pt,ellipse,minimum width=43pt]         (E2) [below of=A2]  {a1,r2,t2};

\path (A1) edge node [swap] {n1} (B2)
      (B2) edge node [swap] {n2} (A2)
      (A2) edge node [swap] {n3} (A1)
      (A1) edge node {m1} (C2)
      (C2) edge node {m2} (A3)
      (A3) edge node {m3} (A1)
      (A3) edge node {n1} (D2)
      (D2) edge node {n2} (A4)
      (A4) edge node {n3} (A3)
      (A2) edge node [swap] {m1} (E2)
      (E2) edge node [swap] {m2} (A4)
      (A4) edge node [swap] {m3} (A2)
      (D2) edge node [swap,xshift=0.6cm,yshift=-0.4cm]{n3} (B2)
      (E2) edge node [xshift=-0.4cm,yshift=-0.35cm]{m3} (C2)
      ;
\end{tikzpicture}
    \caption{Global model of the $\AMAS$ in Fig.~\ref{fig:tramas} with synchronizations on transitions}
    \label{fig:iis}
\end{figure}
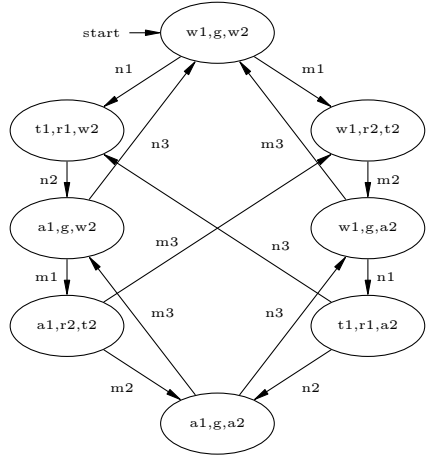
\subsection{Synchronization on data}
\label{sec:mas-pl}
In \cite{LSWW13}, the authors propose another approach to model concurrent 
systems composed by several agents interacting with each other.
In their model, each agent is represented by a \emph{reactive module}.
We add the labeling function with codomain $\Lbl$ to this model and assign the names to the transitions.
\begin{definition}
\label{d:module}
A (labeled) module of an agent $j$ is a tuple $M_j = (X_j, I_j, \Ls_j, \Trext_j, \lambda_j, \iota_j, \ell_j, \Lbl)$, where:
\begin{itemize}
    \item $X_j$ is a finite set of variables controlled by agent $j$; 
    \item $I_j$ is a finite set of variables from which agent $j$ directly depends on. It must hold $I_j \cap X_j = \emptyset$.
    \item $\Ls_j$ is the set of local states;
    \item $\lambda_j: \Ls_j \rightarrow D^{X_j}$ is a function labeling each state $q$ with a valuation $\lambda_j(q): X \rightarrow D$ (with $D$ domain for both internal and external variables).
    \item $\Trext_j: \Ls_j \times D^{I_j} \rightarrow \Ls_j$ is the transition function;
    \item $\iota_j$ is the initial state,
    \item $\ell_j: \Ls_j \times D^{I_j} \times \Ls_j \rightarrow 2^{\Lbl_j}$ is a labeling function such that if $\ell_j(p, v, q) = \ell_j(p', v', q')$, then $\lambda_j(p) = \lambda_j(p')$, $\lambda_j(q) = \lambda_j(q')$, and $\Trext_j(p, v') = q, \Trext_j(p', v) = q$. 
\end{itemize}
The singleton labels of the transitions will be identified with their only elements.
\end{definition}

To store valuations of variables, we use functions (sets/multisets) instead of tuples for mathematical elegance and formula clarity. Labels are necessary for Petri net synthesis and are represented as sets (specifically singletons) for technical reasons, ensuring more elegant definitions of fusions without conflicts or hierarchical component structures. Transition names are typically chosen based on differences in internal variable valuations between input and output states, while label splitting can be performed based on the valuation of external variables.

\begin{example}
\label{ex:tgs-2sem1opt}
The Train-Gate-Controller example can be adapted to this new
semantics. In this case, the agents are 
isomorphic to those in Fig.~\ref{fig:tramas}, but all the 
actions are local, and their occurrence may be constrained to
the value in the local states of other agents.
In the example, we define the sets of variables for the train $1$ on the left as 
$X_1 = \{w1, a1, t1\}$, $I_1 = \{r1\}$.
Since each variable can take two different values, there are eight possible valuations.
However, only three of them are in the image of labeling function $\lambda_1$.
Analogously for train $2$, $X_2 = \{w2, a2, t2\}$, $I_2 = \{r2\}$.
Finally, for the controller $c$, $X_c = \{g, r1, r2\}$, and $I_c =\{w1, w2, a1, a2\}$. 

In particular, train $1$ can execute the
actions $n1$ and $n2$ only if the controller is in the state $r1$.
Formally, transition $n1$ is denoted by $(\{w1=1,t1=0,a1=0\},\{r1=1\},\{w1=0,t1=1,a1=0\})$.
The controller can execute $n1$ only if train $1$ is in the state $w1$,
and $n2$ only if train $1$ is in the state $a1$.
Symmetrically for train $2$.
However, note that controller have four input variables.
Two of them are controlled by train $1$, the other two by train $2$.
This means that, formally, we have four instances of transition $n1$:
\begin{itemize}
    \item 
$(\{g:1,r1:0,r2:0\},\{w1:1,a1:0,w2:0,a2:0\},\{g:0,r1:1,r2:0\})$,
    \item 
$(\{g:1,r1:0,r2:0\},\{w1:1,a1:0,w2:1,a2:0\},\{g:0,r1:1,r2:0\})$,
    \item 
$(\{g:1,r1:0,r2:0\},\{w1:1,a1:0,w2:0,a2:1\},\{g:0,r1:1,r2:0\})$,
    \item 
$(\{g:1,r1:0,r2:0\},\{w1:1,a1:0,w2:1,a2:1\},\{g:0,r1:1,r2:0\})$.
\end{itemize}
By noting that there are no state of train $2$ which is compatible with the last version,
we can erase it.

In some cases, the constraints for the actions to occur may become redundant. Consider for example the action $n2$ performed by train 1: we required that the controller must be in the state $r1$ to allow its occurrence, but this would happen even without imposing it explicitly, since there is no other evaluation of the system in which $n2$ would be allowed (see Fig. \ref{fig:datasync})
\end{example}
\begin{example}
\label{ex:tgs-2sem2opt}
The agents in the system described in Ex. \ref{ex:tgs-2sem1opt} can also be designed with a unique internal variable,
assuming three possible values (following~\cite{MJK23}).
We can deduce this from example \ref{ex:tgs-2sem1opt}, by observing that in the context with three binary variables,
there can never be more than one variable inside the same agent with value $1$;
then, when we can equivalently assume the presence of a single variable taking as
values the three binary variables, and assigning the value of the variable in a state to the variable with value $1$ in the binary case.
With this design, to describe train $i$ we assign $X_1 =\{si\}$, $i\in \{1,2\}$ where
$si$ can assume values $\{wi, ai, ti\}$,
and $I_1 = \{sc\}$, where $sc$ assume values
$\{g, r1, r2\}$; the controller is defined
by $X_c = \{sc\}$, and $I_c = \{s1, s2\}$.
Also in this case, the transition systems of
the modules are isomorphic to those in Fig. \ref{fig:tramas}. 
A main difference between the two proposed designs is that in this second case, the agents can access more information than in 
the option presented in Ex. \ref{ex:tgs-2sem1opt}. For example, in this last setting, a trains can access 
the value of single state of the controller, being able to observe whether the entrance 
to the tunnel has been granted to the other train. This does not happen in the context described in Ex. \ref{ex:tgs-2sem1opt},
where $I_i = \{r_i\}$, therefore the only information given to train $i$ from the controller is whether the access was granted for itself, and not for the other train.
Although having more privacy for the agents
would be preferable in some circumstances,
increasing the number of variables could generate larger systems after synchronization,
as it will be discussed later in the paper,
therefore the choice must be done depending on the specific needs of the application.
Intermediate options (e.g. some agents with binary variables and others with ternary) are also possible to mitigate pros and cons of the two approaches.
\end{example}
We can expand an open system (module) to closed one by taking into account all the possible valuations. 
However, there are two important issues to be fixed.
We need to provide a reasonable method of changing the valuation of external variable and decide which valuation is the initial one.

To address the first issue, we can refer to assume-guarantee reasoning and propose the most liberal approach, allowing for arbitrary changes of a single variable at time.
It is much harder to guess or draw the initial valuation, hence we leave it as a parameter of universal closing.

\begin{definition}
\label{d:cl_module}
A \emph{universal closure} of a module $M = (X, I, \Ls, \Trext, \lambda, \iota, \ell)$ with a valuation $v_{init}\in D^I$ is a system
$M' = (X\cup I, \emptyset, \Ls\times D^I, \Trext', \lambda', \iota', \ell')$, where
\begin{itemize}
    \item $\Trext'=
    \{((x,v),(y,v)) \;|\; \Trext(x,v)=y\}
    \cup
    \{((x,v),(x,w)) \;|\; v(i)\neq w(i) \wedge \forall_{j\neq i} v(j)=w(j)\}
    $;
    \item $\lambda'(x,v)(q) = \lambda(x)(q)$ for $q\in X$
    and $\lambda'(x,v)(q) = v(q)$ otherwise;
    \item $\iota'(q) = (\iota(q),v_{init})$;
    \item $\ell'(x,v,y) = \ell(x,v,y)$ for $y = \Trext(x,v)$ 
    and $\ell'(x,v,y) = \emptyset$ otherwise.
\end{itemize}
\end{definition}

The images of universal closure are called \emph{closed modules}.

\paragraph{Composition of agents}

Having two modules $M_1$ and $M_2$, we can utilize internal variables of one of them which are external variables of the other to aggregate them in a natural way and obtain a larger system. 
We follow the asymmetric part of the composition presented in~\cite{MJK23}.

For this reason we need to recall the notion of compatible valuations.
Let $Y,Z\subseteq X$ and $\rho_1\in D^Y$ while $\rho_2\in D^Z$.
We say that $\rho_1$ is compatible with $\rho_2$ (denoted by $\rho_1 \sim \rho_2$)
if for any $x\in Y\cap Z$ we have $\rho_1(x)=\rho_2(x)$.
We can compute the union of $\rho_1$ with $\rho_2$ which is compatible with $\rho_1$ by setting
$(\rho_1 \cup \rho_2)(x) = \rho_1(x)$ for $x\in Y$ and
$(\rho_1 \cup \rho_2)(x) = \rho_2(x)$ for $x\in Z$.

\begin{definition}
Let $M_1 = (X_1, I_1, \Ls_1, \Trext_1, \lambda_1, \iota_1, \ell_1)$ and $M_2 = (X_2, I_2, \Ls_2, \Trext_2 \lambda_2, \iota_2, \ell_2)$ be two modules with $X_1\cap X_2=\emptyset$ and $\Lbl_1\cap\Lbl_2 = \emptyset$.
We define $M_1 | M_2 = (X=X_1\uplus X_2,I=(I_1\cup I_2)\setminus X,\Ls=\Ls_1\times \Ls_2,\Trext,\lambda,\iota=(\iota_1,\iota_2),\ell)$, where
\begin{itemize}
\item $\lambda:L\to D^X$, $\lambda(q_1,q_2)=\lambda_1(q_1)\cup\lambda_2(q_2)$,
\item $\Trext$ is the minimal transition relation derived by the set of rules presented below:
\[
\bf{ASYN_L}\;\;\;
\begin{array}{c}
q_1\xrightarrow[]{\alpha_1}_{\Trext_1}{q'_1\;\;\;\;q_2\xrightarrow[]{\alpha_2}_{\Trext_2}}{q'_2}\\
\alpha_1\sim\alpha_2 \;\;\;\; \lambda_1(q_1)\sim\alpha_2 \;\;\;\; \lambda_2(q_2)\sim\alpha_1\\
\hline
(q_1,q_2)\xrightarrow[]{\alpha}_\Trext(q'_1,q_2)
\end{array}
\]
\[
\bf{ASYN_R}\;\;\;
\begin{array}{c}
q_1\xrightarrow[]{\alpha_1}_{\Trext_1}{q'_1\;\;\;\;q_2\xrightarrow[]{\alpha_2}_{\Trext_2}}{q'_2}\\
\alpha_1\sim\alpha_2 \;\;\;\; \lambda_1(q_1)\sim\alpha_2 \;\;\;\; \lambda_2(q_2)\sim\alpha_1\\
\hline
(q_1,q_2)\xrightarrow[]{\alpha}_\Trext(q_1,q'_2)
\end{array}
\]
where $\alpha\in D^I$ and $\alpha(x) = \alpha_1(x)\cup\alpha_2(x)$.
\item $\ell((q_1,q_2),(\alpha_1\cup\alpha_2)\setminus X,(q'_1,q_2))) = \ell(q_1,\alpha_1,q'_1)$,
and $\ell((q_1,q_2),(\alpha_1\cup\alpha_2)\setminus X,(q_1,q'_2))) = \ell(q_2,\alpha_2,q'_2)$
\end{itemize}
\end{definition}

Referring to~\cite{MJK23}, we can compose in this way more than two modules (and such an operation is associative).
Note that the set of external variables of a composition might reduce to the empty set and we call the compositions of modules (or single modules) with non-empty sets of external variables open (multi-agent) systems, while those with empty set of external variables - closed (multi-agent) systems.

\begin{remark}
    The composition of closed modules is a closed module.
\end{remark}

Similarly to the observation we made in Sec. \ref{sec:mas-tr}, the closed multi-agent system obtained through data synchronization may have unreachable parts.

\begin{example}
The global model based on these semantics 
is shown in Fig.~\ref{fig:datasync}. 
Note that this is valid for both possible choices of presented in Examples~\ref{ex:tgs-2sem1opt} and~\ref{ex:tgs-2sem2opt}. 
The labels used as names of the states are the names of the only variables with value $1$ in the first case or the values of the only variable in the second case.
Comparing it with the model in Fig.~\ref{fig:datasync}
one can see that the behaviors allowed in the two models differs,
although both of them satisfy the requirement of having only
one train in the tunnel at any moment. For example, in the model
in Fig.~\ref{fig:datasync}, once train $1$ has obtained the 
permission to enter, if it is fast enough, it can enter infinitely
often in the tunnel without having this permission removed.
Hence, in order to guarantee that both trains can eventually
enter the tunnel if they enter their waiting state, the
cooperation of the trains is required.
In Fig.~\ref{fig:iis}, this must be guaranteed by the controller alone.
\end{example}
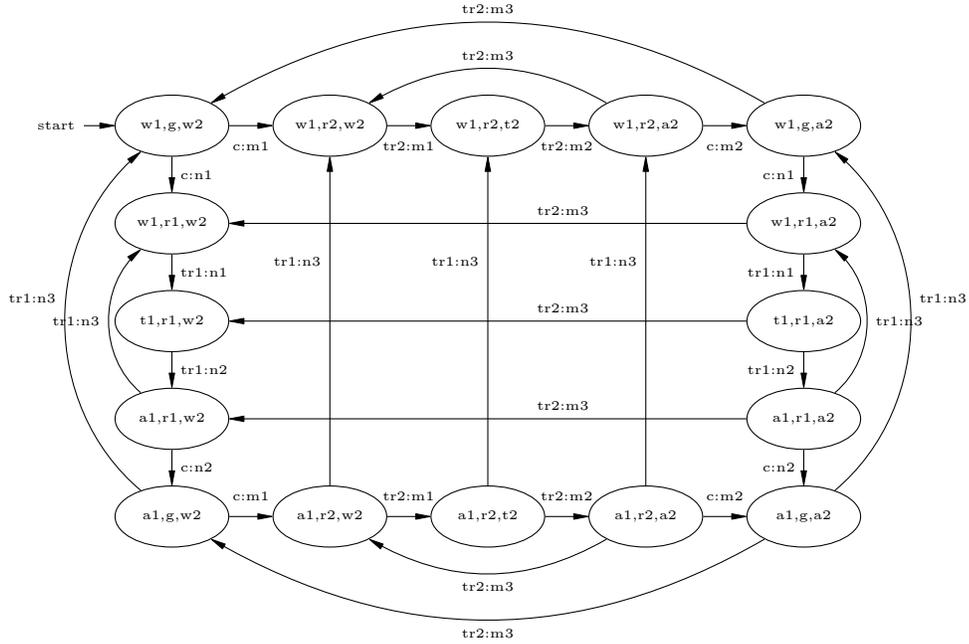
\begin{figure}[ht]
    \centering
   \begin{tikzpicture}[->,auto,>=arrow30,node distance=1.3cm,font=\tiny]

\node[initial,state,inner sep=-1pt,ellipse,minimum width=43pt] (A1)                {w1,g,w2};
\node[state,inner sep=-1pt,ellipse,minimum width=43pt]         (B1) [below of=A1]  {w1,r1,w2};
\node[state,inner sep=-1pt,ellipse,minimum width=43pt]         (B2) [below of=B1]  {t1,r1,w2};
\node[state,inner sep=-1pt,ellipse,minimum width=43pt]         (B3) [below of=B2]  {a1,r1,w2};
\node[state,inner sep=-1pt,ellipse,minimum width=43pt]         (A2) [below of=B3]  {a1,g,w2};
\node[state,inner sep=-1pt,ellipse,minimum width=43pt]         (C1) [right of=A1,xshift=8mm]  {w1,r2,w2};
\node[state,inner sep=-1pt,ellipse,minimum width=43pt]         (C2) [right of=C1,xshift=8mm]  {w1,r2,t2};
\node[state,inner sep=-1pt,ellipse,minimum width=43pt]         (C3) [right of=C2,xshift=8mm]  {w1,r2,a2};
\node[state,inner sep=-1pt,ellipse,minimum width=43pt]         (A3) [right of=C3,xshift=8mm]  {w1,g,a2};
\node[state,inner sep=-1pt,ellipse,minimum width=43pt]         (D1) [below of=A3]  {w1,r1,a2};
\node[state,inner sep=-1pt,ellipse,minimum width=43pt]         (D2) [below of=D1]  {t1,r1,a2};
\node[state,inner sep=-1pt,ellipse,minimum width=43pt]         (D3) [below of=D2]  {a1,r1,a2};
\node[state,inner sep=-1pt,ellipse,minimum width=43pt]         (A4) [below of=D3]  {a1,g,a2};
\node[state,inner sep=-1pt,ellipse,minimum width=43pt]         (E1) [right of=A2,xshift=8mm]  {a1,r2,w2};
\node[state,inner sep=-1pt,ellipse,minimum width=43pt]         (E2) [right of=E1,xshift=8mm]  {a1,r2,t2};
\node[state,inner sep=-1pt,ellipse,minimum width=43pt]         (E3) [right of=E2,xshift=8mm]  {a1,r2,a2};

\path (A1) edge node {c:n1} (B1)
      (B1) edge node {tr1:n1} (B2)
      (B2) edge node {tr1:n2} (B3)
      (B3) edge node {c:n2} (A2)
      (A2) edge [bend left,out=50, in=130] node [yshift=3mm] {tr1:n3} (A1)
      (B3) edge [bend left,out=50, in=130] node {tr1:n3} (B1)
      (A1) edge node [swap,yshift=-1mm] {c:m1} (C1)
      (C1) edge node [swap,yshift=-1mm] {tr2:m1} (C2)
      (C2) edge node [swap,yshift=-1mm] {tr2:m2} (C3)
      (C3) edge node [swap,yshift=-1mm] {c:m2} (A3)
      (A3) edge [bend right] node [swap] {tr2:m3} (A1)
      (C3) edge [bend right] node [swap] {tr2:m3} (C1)
      (A3) edge node [swap] {c:n1} (D1)
      (D1) edge node [swap] {tr1:n1} (D2)
      (D2) edge node [swap] {tr1:n2} (D3)
      (D3) edge node [swap] {c:n2} (A4)
      (A4) edge [bend right,out=-50, in=-130] node [swap,yshift=3mm] {tr1:n3} (A3)
      (D3) edge [bend right,out=-50, in=-130] node [swap] {tr1:n3} (D1)
      (A2) edge node [yshift=1mm] {c:m1} (E1)
      (E1) edge node [yshift=1mm] {tr2:m1} (E2)
      (E2) edge node [yshift=1mm] {tr2:m2} (E3)
      (E3) edge node [yshift=1mm] {c:m2} (A4)
      (A4) edge [bend left] node {tr2:m3} (A2)
      (E3) edge [bend left] node {tr2:m3} (E1)
      (E1) edge node [yshift=8mm] {tr1:n3} (C1)
      (E2) edge node [yshift=8mm] {tr1:n3} (C2)
      (E3) edge node [yshift=8mm] {tr1:n3} (C3)
      (D1) edge node [swap,xshift=1cm] {tr2:m3} (B1)
      (D2) edge node [swap,xshift=1cm] {tr2:m3} (B2)
      (D3) edge node [swap,xshift=1cm] {tr2:m3} (B3)
      ;

\end{tikzpicture}
    \vspace*{-1cm}
    \caption{Global model of the TGC MAS with two trains and 
    synchronizations on data. Since the transitions are local,
    but some have the same label, the owner is specified before 
    their label.}
    \label{fig:datasync}
\end{figure}

\begin{remark}
    Note that Fig.3 is in fact very similar to Fig.2. If in the case of the synchronization on data we force to execute $trx:(m/n)y$ just after $c:(m/n)y$ (for $y\in\{1,2\}$) treating such pairs as synchronized actions, and allow $trx:(m/n)3$ in all reachable (by forced synchronized actions) states we get a transition system for global model isomorphic with the one obtained for the synchronization on transitions. 
\end{remark}
     
\section{1-safe Petri nets with read and inhibitor arcs}
\label{sec:pn}
We can model both semantics defined in Sec. \ref{sec:mas-ts} in the framework of $1$-safe Petri nets.
In this section we provide the formal definitions of the Petri nets that we use for the modeling, we recall region theory for the synthesis of transition systems (Sec. \ref{sec:synt}) and we show how region theory can be used to obtain $1$-safe Petri nets for the two models defined in the previous section (Sec. \ref{sec:synt-mas}).

Petri nets were introduced by Carl Adam Petri in his PhD thesis \cite{Petri62} as a 
formal graphical model to represent and analyze concurrent systems. 
In this section we provide some basic definitions that will be useful in the rest of the 
paper. For an extensive overview about Petri nets and their applications we refer to \cite{murata89,peterson77}. 

A \emph{plain} net is characterized by a set of \emph{places} or \emph{conditions} $P$, represented as circles, by a set of \emph{transitions} $T$, represented as squares, 
and by a  \emph{flow relation} between them 
$\flow\subseteq (P\times T) \cup (T \times P)$, represented with arcs. 
Although sharing the same name,
transitions in Petri nets should not be confused with transitions in transition systems.

For each element $x\in P\cup T$, its preset is $^\bullet x = \{y\in P\cup T : (y, x)\in \flow\}$, and its postset is $x^\bullet = \{y\in P\cup T : (x, y)\in F\}$. 
For each transition $t\in T$, we assume that its preset and its postset are non-empty, 
i.e. $^\bullet t \neq \emptyset$ and $t^\bullet \neq \emptyset$. 
The elements in $^\bullet t$ are also called \emph{preconditions} of $t$, and the elements in $t^\bullet$ 
are also called \emph{postconditions}.

A net system is a quadruple $\Sigma = (P, T, F, m_0)$, where $P, T, F$ are the elements of the net, 
and $m_0: P \rightarrow \mathbb{N}$ is the \emph{initial marking}. 
A transition $t\in T$ is enabled in a marking $m$ if for each $p\in {^\bullet t}$, $m(p) \geq 0$. 
If a transition is enabled, it can \emph{occur} or \emph{fire}, and its occurrence generates a 
new marking $m'$ defined as follows. 
\[
    m'(p) =
  \begin{cases}
    m(p)-1 & \text{for all \(p \in {^\bullet t} \setminus t^\bullet\)}\\
    m(p)+1 & \text{for all \(p \in t^\bullet \setminus {^\bullet t}\)}\\
    m(p)   & \text{in all other cases.}
  \end{cases}
\]
In symbols, $m[t\rangle$ denotes that $t$ is enabled in $m$, while $m[t\rangle m'$ denotes 
that $m'$ is the marking produced from the occurrence of $t$ in $m$. 
A marking $m$ is reachable in a net system ${\Sigma = (P, T, F, m_0)}$ if 
there is a sequence of transitions (called a \emph{firing sequence}) 
$t_1...t_n$ such that $m_0[t_1\rangle m_1...m_{n-1}[t_n\rangle m$. 
The set of all the reachable markings is denoted with $[m_0\rangle$. 
A transition $t\in T$ is 1-live if there is a marking $m\in [m_0\rangle$ such that $m[t\rangle$. 

Let $m$ be a reachable marking, and $t_1, t_2\in T$ be two transitions enabled in $m$: 
$t_1$ and $t_2$ are in \emph{conflict} in $m$ if $^\bullet t_1 \cap {^\bullet t_2} \neq \emptyset$; 
$t_1$ and $t_2$ are concurrent in $m$ if $^\bullet t_1 \cap {^\bullet t_2} = \emptyset$ and $t_1^\bullet \cap t_2 ^\bullet = \emptyset$. 

In this paper we work with the class of 1-safe net systems. 
A net system is \emph{1-safe} if, for each $m\in [m_0\rangle$ and for each $p\in P$, $m(p) \leq 1$. 
In a 1-safe system, each marking can (and will) be considered as a set of places, and each place can 
be interpreted as a proposition, that is true if the place belongs to the marking, and false otherwise. 
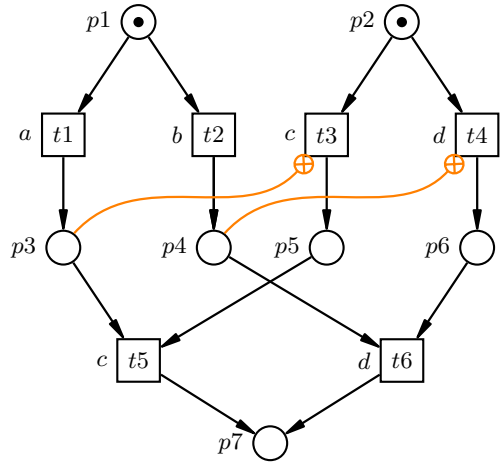
\begin{figure}
    \centering
        \begin{tikzpicture}[->,node distance=1.5cm,>=arrow30,line  width=0.3mm,scale=1.1,auto,bend angle=45,font=\small]
     \tikzstyle{place}=[draw,circle,thick,minimum size=4.5mm]
     \tikzstyle{transition}=
                [draw,regular polygon,thick,
                 regular polygon sides=4,minimum size=8mm, inner sep = -2pt]

    \node[place,draw=black,tokens=1,label=left:$p1$,right of=A2](p1){};
    \node[place,draw=black,tokens=1,label=left:$p2$,right of= p1, xshift=2cm](p2){};
    \node[transition,draw=black,label=left:$a$, below of= p1, xshift=-1cm](t1){$t1$};
    \node[transition,draw=black,label=left:$b$,below of= p1, xshift=1cm](t2){$t2$};
    \node[transition,draw=black,label=left:$c$,below of= p2,xshift=-1cm](t3){$t3$};
    \node[transition,draw=black,label=left:$d$,below of= p2,xshift=1cm](t4){$t4$};
    \node[place,draw=black,tokens=0,label=left:$p3$,below of= t1](p3){};
    \node[place,draw=black,tokens=0,label=left:$p4$,below of= t2](p4){};
    \node[place,draw=black,tokens=0,label=left:$p5$,below of= t3](p5){};
    \node[place,draw=black,tokens=0,label=left:$p6$,below of= t4](p6){};
    \node[transition,draw=black,label=left:$c$,below of= p4,xshift=-1cm](t5){$t5$};
    \node[transition,draw=black,label=left:$d$,below of= p6,xshift=-1cm](t6){$t6$};
    \node[place,draw=black,tokens=0,label=left:$p7$,below of= t5, xshift=1.75cm, yshift=0.4cm](p7){};

    \path (p1) edge [->] (t1);
    \path (p1) edge [->] (t2);
    \path (t1) edge [->] (p3);
    \path (t2) edge [->] (p4);
    \path (t2) edge [->] (p4);
    \path (p2) edge [->] (t3);
    \path (p2) edge [->] (t4);
    \path (t3) edge [->] (p5);
    \path (t4) edge [->] (p6);
    \path (p3) edge [->] (t5);
    \path (p5) edge [->] (t5);
    \path (p4) edge [->] (t6);
    \path (p6) edge [->] (t6);
    \path (t5) edge [->] (p7);
    \path (t6) edge [->] (p7);

    \path [->,>=read] (p3) [color=orange, bend left,out=30,in=-150,min distance=1cm]  edge (t3);
    \path [->,>=read] (p4) edge [color=orange, bend left,out=30,in=-150,min distance=1cm] (t4);
    
     \end{tikzpicture}
    \caption{Example of $1$-safe labeled net system with read arcs.}
    \label{fig:pn-ex}
\end{figure}    
A \emph{labeled} Petri net $\Sigma_\ell= (\Sigma, \ell)$ is a 1-safe net system with a function 
$\ell: T \rightarrow 2^\Lambda$, where $\Lambda$ is a set of labels. 
In the cases where images of labeling function are singletons, we will use the only elements instead of sets consisting of one label.
Abusing the notation, for each 
$T' \subseteq T$ subset of $T$, we will denote with $\ell(T') = \{\alpha\in 2^\Lambda : \exists t\in T': \ell(t) = \alpha\}$ 
the set of labels of the elements in $T'$. The set $\ell(T)$ is the \emph{alphabet} of $\Sigma_\ell$.
When $\ell$ is clear from the context we will refer to the labeled net only as $\Sigma$.
\begin{example}
Consider the net system in Fig. \ref{fig:pn-ex}, for the moment ignoring the orange arcs.
It is a labeled and $1$-safe net system, with initial marking $\{p1, p2\}$.
The transitions enabled in the initial marking are $\{t1, t2, t3, t4\}$.
Transitions $t1$ and $t2$ are in conflict with each other, whereas $t1$ and $t3$ are concurrent.
The net system is labeled, therefore each transitions is represented both with a unique identifier ($t_i$, $i\in \{1,...,6\}$)
and with its label; some labels are shared, e.g. $\ell(t5) = \ell(t6) = c$.
\end{example}
The sequential behavior of a labeled Petri net 
can be described by an initialized labeled transition system, 
where each state corresponds to a reachable marking, 
and each arc is labeled by the label of the transition leading from the source marking to the target one.

\begin{definition} 
	Let $\Sigma = (P, T, F, m_0, \ell)$ be a  labeled Petri net, 
	its \emph{marking graph} is a transition system $\MG(\Sigma) = ([m_0\rangle, \ell(T), Ar, m_0)$
	where  $Ar= \{ f: [m_0\rangle \times \ell(T) \rightarrow [m_0\rangle \mid f(m, \ell(t)) = m' \iff  \,   m[t \rangle m' \}$. 
\end{definition}

Let $\Sigma = (P, T, \flow, m_0, \ell)$ be a labeled $1$-safe system.
A subsystem of $\Sigma$ is a labeled $1$-safe system $\Sigma' = (P', T', \flow', m'_0, \ell')$ such that 
$P' \subseteq P$, $T' \subseteq T$, $m'_0 \subseteq m_0$, $\flow' \equiv \flow_{|(P' \times T')\cup (T' \times P')}$, and $\ell' \equiv \ell_{|T'}$.
Let $P_s$ be a subset of places, and $T_s = \{t\in T : \exists p\in P_s : t\in \pre{p} \lor t \in \post{p}\}$
the set of transitions for which an element of $P_s$ is a precondition or a postcondition. 
The subsystem $\Sigma_s = (P_s, T_s, \flow_s, m_{0,s}, \ell)$ is defined as \emph{subsystem generated by $P_s$}.
This subsystem is a \emph{sequential component} of a $1$-safe net system $\Sigma$ iff for each $t \in T_s$, $|\pre{t}| = |\post{t}| = 1$, and for each $m \in [m_{0,s}\rangle$, $|m| = 1$.

A net system with read arcs (inhibitor arcs) is a tuple $\Sigma^i = (P,T,\flow,\read,m_0)$ (respectively $\Sigma^r = (P,T,\flow,\inh,m_0)$), where $P, T, \flow$ are elements of the net and $m_0$ is the initial marking, while 
$\read\subseteq P\times T$ is an \emph{activation relation} ($\inh\subseteq P\times T$ is an \emph{inhibition relation}, respectively).

For each element $t\in T$, 
the set of its \emph{activators} is $^\otimes t = \{p\in P: (p,t)\in \read\}$, 
while 
the set of its \emph{inhibitors} is $^\circ t = \{p\in P: (p,t)\in \inh\}$.

A transition $t$ is enabled in a marking $m$ of a system with activators $\Sigma^r$ 
if for each $p\in {^\bullet t} \cup {^\otimes t}$, $m(p)\geq 0$.
Similarly, $t$ is enabled in a marking $m$ of a system with inhibitors $\Sigma^i$ 
if for each $p\in {^\bullet t}$ not only $m(p)\geq 0$, but also $m(p')=0$ 
for every $p'\in {^\circ t}$. 
For both systems with inhibitor arcs and systems with read arcs, \emph{firing rules} for enabled transitions are the same as in systems without inhibition and activation relations. 

Graphically, 
activators are represented by arcs with small filled circles as arrowheads, while 
inhibitors are represented by arcs with small empty circles as arrowheads.

\begin{example}
In Fig. \ref{fig:pn-ex}, the read arcs are represented with orange.
When we consider read arcs, only $t1$ and $t2$ are enabled in the initial marking,
since the condition required by the read arc for the occurrence of $t3$ and $t4$
is not satisfied.
\end{example}

In what follows we provide few notions of fusion operation that merges disconnected nets into one system. 
We will concentrate on the fusion of two nets 
$\Sigma_1 = (P_1, T_1, \flow_1, \read_1, m_{0,1}, \ell_1)$ and
$\Sigma_2 = (P_2, T_2, \flow_2, \read_2, m_{0,2}, \ell_2)$ into a net 
$\Sigma = (P, T, \flow, \read, m_{0}, \ell)$.

\paragraph{Fusion of transitions}
The first idea is to identify copies of the same transition that are present in models of different parts of a system. 
Additionally, we assume that $P_1 \cap P_2 = \emptyset$.
Since we operate on labeled Petri nets, we would like to identify all transitions with the same label, which multiplies the number of equilabeled transitions in the resulting system.

Formally, we will denote $\Sigma$ as $\Sigma_1 \bowtie \Sigma_2$, where:
\begin{itemize}
    \item $P = P_1 \uplus P_2$;
    \item $T = 
    \{(t_1,t_2)\in T_1\times T_2 \;|\; \ell_1(t_1)=\ell_2(t_2)\} 
    \cup 
    \{(t,\varepsilon) \;|\; t\in T_1 \wedge \forall_{t'\in T_2}\;\ell_1(t)\neq \ell_2(t')\} 
    \cup 
    \{(\varepsilon,t) \;|\; t\in T_2 \wedge \forall_{t'\in T_1}\;\ell_2(t)\neq \ell_1(t')\}$;
    \item $\flow = 
    \{(p,(t_1,t_2)) \;|\; (p,t_1)\in\flow_1\}
    \cup
    \{(p,(t_1,t_2)) \;|\; (p,t_2)\in\flow_2\}
    \cup
    \{((t_1,t_2),p) \;|\; (t_1,p)\in\flow_1\}
    \cup
    \{((t_1,t_2),p) \;|\; (t_2,p)\in\flow_2\}$;
    \item $\read = 
    \{(p,(t_1,t_2)) \;|\; (p,t_1)\in\read_1\}
    \cup
    \{(p,(t_1,t_2)) \;|\; (p,t_2)\in\read_2\}$;
    \item $m_0 = m_{0,1} \cup m_{0,2}$;
    \item $\ell(t_1,t_2) = \ell_1(t_1)\cup\ell_2(t_2)$.
 \end{itemize}
Fig.~\ref{fig:synctr-pn} show an example of synchronization on
transitions, where two Petri nets synchronize on transitions
labeled with $a$.
\begin{figure}
    \centering
        \begin{tikzpicture}[node distance=1.3cm,>=arrow30,line  width=0.3mm,scale=1.1,auto,bend angle=45,font=\tiny]
     \tikzstyle{place}=[draw,circle,thick,minimum size=4.5mm]
     \tikzstyle{transition}=
                [draw,regular polygon,thick,
                 regular polygon sides=4,minimum size=8mm, inner sep = -2pt]

     \node(p1)at(3.5,4)[color=red,place,tokens=1,label=above:$p1$]{};
     \node(p2)at(2,2)[color=red,place,tokens=0,label=left:$p2$]{};
     \node(p3)at(5,2)[color=red,place,tokens=0,label=right:$p3$]{};
     \node(a1)at(2,3)[color=red,transition] {$a$};
     \node(a2)at(5,3)[color=red,transition] {$a$};
     \node(b)at(3,3)[color=red,transition] {$b$};
     \node(c)at(4,3)[color=red,transition] {$c$};
     \path(p1) edge [->] (a1);
     \path(p1) edge [->] (a2);
     \path(a1) edge [->] (p2);
     \path(a2) edge [->] (p3);
     \path(p2) edge [->] (b);
     \path(b) edge [->] (p1);
     \path(p3) edge [->] (c);
     \path(c) edge [->] (p1);

     \node(p4)at(7,4)[color=blue,place,tokens=1,label=right:$p4$]{};
     \node(p5)at(7,2)[color=blue,place,tokens=0,label=left:$p5$]{};
     \node(a3)at(6,3)[color=blue,transition] {$a$};
     \node(d)at(8,3)[color=blue,transition] {$d$};
     \path(p4) edge [->] (a3);
     \path(a3) edge [->] (p5);
     \path(p5) edge [->] (d);
     \path(d) edge [->] (p4);

     \node(p1s)at(10.5,5)[color=red,place,tokens=1,label=above:$p1$]{};
     \node(p2s)at(9,2)[color=red,place,tokens=0,label=left:$p2$]{};
     \node(p3s)at(12,2)[color=red,place,tokens=0,label=right:$p3$]{};
     \node(a1s)at(9,3)[color=red,transition] {$a$};
     \node(a2s)at(12,3)[color=red,transition] {$a$};
     \node(bs)at(10,3)[color=red,transition] {$b$};
     \node(cs)at(11,3)[color=red,transition] {$c$};
     \node(p4s)at(10.5,4)[color=blue,place,tokens=1,label=right:$p4$]{};  
     \node(p5s)at(10.5,1)[color=blue,place,tokens=0,label=right:$p4$]{};  
     \node(a31s)at(9.1,3)[color=blue,transition] {$a$};
     \node(a32s)at(11.9,3)[color=blue,transition] {$a$};
     \node(ds)at(10.5,2.1)[color=blue,transition] {$d$};
     \path(p1s) edge [->] (a1s);
     \path(p1s) edge [->] (a2s);
     \path(a1s) edge [->] (p2s);
     \path(a2s) edge [->] (p3s);
     \path(p2s) edge [->] (bs);
     \path(bs) edge [->] (p1s);
     \path(p3s) edge [->] (cs);
     \path(cs) edge [->] (p1s);
     \path(p4s) edge [->] (a31s);
     \path(p4s) edge [->] (a32s);
     \path(a1s) edge [->] (p5s);
     \path(a2s) edge [->] (p5s);
     \path(p5s) edge [->] (ds);
     \path(ds) edge [->] (p4s);
     
     \end{tikzpicture}
    \caption{Simple example of synchronization on transitions.
    Let $\Sigma_1$ be the net on the left and $\Sigma_2$ be the
    net in the center; the net on the right represent their
    synchronization $\Sigma_1 \bowtie \Sigma_2$.
    }
    \label{fig:synctr-pn}
\end{figure}
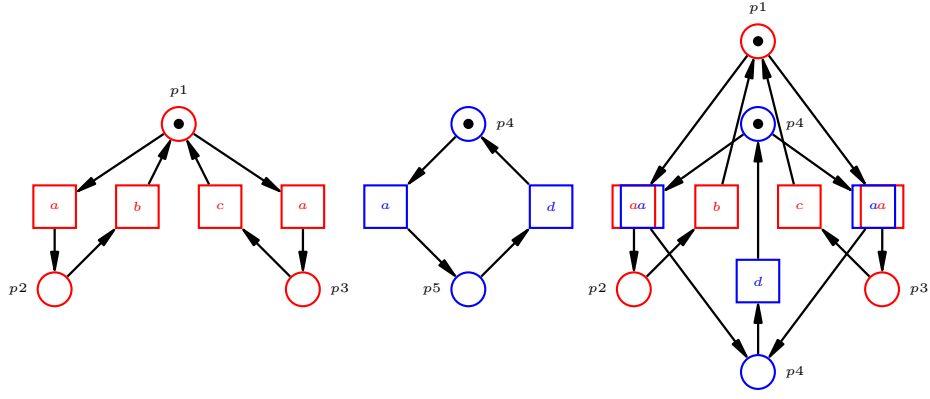
\paragraph{Fusion of places}
Similarly, we can identify corresponding places that are present in models of different parts of a system.
Note that this time the fusion is simpler, as places are not labeled.
We naturally assume that $T_1 \cap T_2 = \emptyset$ (while there may by many transitions with the same label in both parts), and that for every $p \in P_1 \cap P_2$, $m_{0,1}(p) = m_{0,2}(p)$.

Formally, we will denote $\Sigma$ as $\Sigma_1 \infty \Sigma_2$, where:
\begin{itemize}
    \item $P = P_1 \cup P_2$; 
    \item $T = T_1 \uplus T_2$;
    \item $\flow = \flow_1 \cup \flow_2$;
    \item $\read = \read_1 \cup \read_2$;
    \item $m_0(p) = m_{0,1}(p)$ if $p\in P_1$ 
    and $m_0(p) = m_{0,2}(p)$ otherwise;
    \item $\ell(t) = \ell_1(t)$ if $t\in T_1$
    and $\ell(t) = \ell_2(t)$ otherwise.
 \end{itemize}
\paragraph{Fusion of sequential components}
Finally, we identify families of pairwise disjoint sequential components.
For this purpose, we need to choose four families of subsets $X_1=\{X^i_1\subseteq P_1 | i=1..f_1\}$
and $X'_1=\{X'^i_1\subseteq P_1 | i=1..f'_1\}$ and $X_2=\{X^i_2\subseteq P_2 | i=1..f_2\}$ and $X'_2=\{X'^i_2\subseteq P_2 | i=1..f'_2\}$ identifying sequential components and bijections $\rho_i:X^i_1\to X^i_2$ and $\rho'_i:X'^i_2\to X'^i_1$ such that
for every $p\in X^i_1$ we have $m_{0,1}(p) = m_{0,2}(\rho_i(p))$ and 
for every $q\in X'^i_2$ we have $m_{0,2}(q) = m_{0,1}(\rho'_i(q))$.
To make the situation completely symmetric, we make use of the set character of labeling functions.

Formally, we will denote $\Sigma$ as $\Sigma_1 \propto_\rho \Sigma_2$, where:
\begin{itemize}
    \item $P = (P_1 \cup P_2)\setminus (\bigcup_i X^i_2 \setminus \bigcup_i X'^i_2) \setminus (\bigcup_i X'^i_1\setminus \bigcup_i X^i_1)$;
    \item $T = 
    \{(\{t_1\},U^{p,q})\in 2^{T_1}\times 2^{T_2}\;|\;\emptyset \neq U^{p,q}=\bigcup_{i\;|\; \{p\}=X_1^i\cap {\bullet t_1}\land \{q\}=X_1^i\cap{t_1\bullet}}\{t_2 \;|\;\rho_i(p)\in{\bullet t_2}\land \rho_i(q)\in{t_2\bullet}\}\}
    \cup
    \{(U^{p,q},\{t_2\})\in 2^{T_1}\times 2^{T_2}\;|\;\emptyset \neq U^{p,q}=$ $ \bigcup_{i\;|\; \{p\}=X'^i_2\cap {\bullet t_2}\land \{q\}=X'^i_2\cap{t_2\bullet}}\{t_1 \;|\;\rho'_i(p)\in{\bullet t_1}\land \rho'_i(q)\in{t_1\bullet}\}\}
    \cup
    \{(\{t\},\emptyset) \;|\; t\in T_1 \wedge ({\bullet t}\cap (\bigcup_i X^i_1 \cup \bigcup_j X'^j_1) = \emptyset \vee {t\bullet}\cap (\bigcup_i X^i_1 \cup \bigcup_j X'^j_1) = \emptyset)\} 
    \cup 
    \{(\emptyset,\{t\}) \;|\; t\in T_2 \wedge ({\bullet t}\cap (\bigcup_i X^i_2 \cup \bigcup_j X'^j_2) = \emptyset \vee {t\bullet}\cap (\bigcup_i X^i_2 \cup \bigcup_j X'^j_2) = \emptyset)\}.
    $
    \item $\flow =
    \{(p, (U_1, U_2)) \; | \; U_1 \subseteq T_1, U_2 \subseteq T_2 \land (p,t_1) \in \flow_1 \land t_1\in U_1\} \cup 
    \{(p, (U_1, U_2)) \; | \;  U_1 \subseteq T_1, U_2 \subseteq T_2 \land (p,t_2) \in \flow_2 \land t_2\in U_2\} \cup
    \{((U_1, U_2), p) \; | \;  U_1 \subseteq T_1, U_2 \subseteq T_2 \land (t_1,p) \in \flow_1 \land t_1\in U_1\} \cup 
    \{((U_1, U_2), p) \; | \;  U_1 \subseteq T_1, U_2 \subseteq T_2 \land (t_2,p) \in \flow_2 \land t_2\in U_2\};    
    $
    \item $\read =
    \{(p, (U_1, U_2)) \; | \; U_1 \subseteq T_1, U_2 \subseteq T_2 \land (p,t_1) \in \read_1 \land t_1\in U_1\} 
    \cup
    \{(p, (U_1, U_2)) \; | \;  U_1 \subseteq T_1, U_2 \subseteq T_2 \land (p,t_2) \in \read_2 \land t_2\in U_2\};
    $
    \item $m_0(p) = m_{0,1}(p)$ if $p\in P_1$ 
    and $m_0(p) = m_{0,2}(p)$ otherwise;
    \item $\ell((U_1,U_2)) = \bigcup_{t_1\in U_1}\ell_1(t_1)\cup\bigcup_{t_2\in U_2}\ell_2(t_2)$ taking a union over empty set equal to $\emptyset$;
 \end{itemize}
The fusion of sequential components can generate different synchronizations, 
depending on the families of places in which we choose to synchronize the agents.
This is illustrated in the following example.
\begin{example}
Consider two nets presented in Fig.~\ref{fig:sync_seq} (on the left). On the right part of that figure one can see two possible synchronizations on sequential components (on the right) between two agents (on the left).
Let $P_1$ be the set of places for the red agent and $P_2$ be the set of places for the blue agent.
For the first composition (shown above in Fig. \ref{fig:sync_seq}), we define $X_1 = \{\{p_1, p_2, p_3\}\}$, $X_2 = \{\{p_4, p_5, p_6\}\}$, 
$X'_1 = X'_2 = \emptyset$, and we define $\rho(p_1) = p_4$, $\rho(p_2) = p_5$, $\rho(p3)=p_6$.
In the second composition (lower part on the right), we connect different components
$\{X_1 =\{\{p_1, p_2\}\}, X_2 = \{\{p_4, p_6\}\}\}$ and $X'_1 = X'_2 = \emptyset$.
The bijection $\rho$ is defined as follows: $\rho(p1) = p4$, $\rho(p2) = p6$.
From this example, we can see that the choice of the components for the synchronization
can significantly change the shape and properties of the composed net. 

\end{example}
\begin{figure}
    \centering
        \begin{tikzpicture}[node distance=1.3cm,>=arrow30,line  width=0.3mm,scale=1.1,auto,bend angle=45,font=\tiny]
     \tikzstyle{place}=[draw,circle,thick,minimum size=4.5mm]
     \tikzstyle{transition}=
                [draw,regular polygon,thick,
                 regular polygon sides=4,minimum size=8mm, inner sep = -2pt]

     \node(w1)at(2,4)[color=red,place,tokens=1,label=above:$p1$]{};
     \node(a1)at(2,1)[color=red,place,tokens=0,label=left:$p3$]{};
     \node(t1)at(3.5,2.5)[color=red,place,tokens=0,label=right:$p2$]{};
     \node(n1)at(3.5,4)[color=red,transition] {$a$};
     \node(n2)at(3.5,1)[color=red,transition] {$b$};
     \node(n3)at(2,2.5)[color=red,transition] {$c$};
     \node(n4)at(2.75,3.25)[color=red,transition] {$d$};     
     \path(n3) edge [->] (w1);
     \path(w1) edge [->] (n1);
     \path(n1) edge [->] (t1);
     \path(t1) edge [->] (n2);
     \path(n2) edge [->] (a1);
     \path(a1) edge [->] (n3);
     \path(t1) edge [->] (n4);
     \path(n4) edge [->] (w1);

     \node(w2)at(4.5,4)[color=blue,place,tokens=1,label=above:$p4$]{};     
     \node(a2)at(4.5,1)[color=blue,place,tokens=0,label=left:$p6$]{};
     \node(t2)at(6,2.5)[color=blue,place,tokens=0,label=right:$p5$]{};
     \node(m1)at(6,4)[color=blue,transition] {$e$};
     \node(m2)at(6,1)[color=blue,transition] {$f$};
     \node(m3)at(4.5,2.5)[color=blue,transition] {$g$};
     \node(m4)at(5.25,2.5)[color=blue,transition] {$h$};
     \node(m5)at(6.65,0.35)[color=blue,transition] {$k$};
     \path(m3) edge [->] (w2);
     \path(w2) edge [->] (m1);
     \path(m1) edge [->] (t2);
     \path(t2) edge [->] (m2);
     \path(m2) edge [->] (a2);
     \path(a2) edge [->] (m3);
     \path(w2) edge [bend left] [->] (m4);
     \path(m4) edge [bend left] [->] (a2);
     \path(a2) edge [bend right] [->] (m5);
     \path(m5) edge [->] [bend right] (t2);

     \node(s1)at(9.5,5)[color=gray,place,tokens=1,label=left:$p1$]{};     
     \node(s3)at(9.5,3)[color=gray,place,tokens=0,label=left:$p3$]{};
     \node(s2)at(11,4)[color=gray,place,tokens=0,label=right:$p2$]{};
     \node(tr1)at(11,5)[color=gray,transition] {$\{e,a\}$};
     \node(tr2)at(11,3)[color=gray,transition] {$\{f,b\}$};
     \node(tr3)at(9.5,4)[color=gray,transition] {$\{g,c\}$};

     \path(tr3) edge [->] (s1);
     \path(s1) edge [->] (tr1);
     \path(tr1) edge [->] (s2);
     \path(s2) edge [->] (tr2);
     \path(tr2) edge [->] (s3);
     \path(s3) edge [->] (tr3);

     \node(s'1)at(9.5,2)[color=gray,place,tokens=1,label=above:$p1$]{};     
     \node(s'5)at(9.5,-1)[color=gray,place,tokens=0,label=below:$p5$]{};
     \node(s'2)at(11,0.5)[color=gray,place,tokens=0,label=right:$p2$]{};
     \node(s'3)at(8,0.5)[color=gray,place,tokens=0,label=left:$p3$]{};
     \node(tr'6)at(11,-1)[color=gray,transition] {$f$};
     \node(tr'7)at(11.65,-1.65)[color=gray,transition] {$k$};
     \node(tr'5)at(11,2)[color=gray,transition] {$e$};
     \node(tr'2)at(9.5,.5)[color=gray,transition] {$\{d,g\}$};
     \node(tr'1)at(10.25,.5)[color=gray,transition] {$\{a,h\}$};
     \node(tr'4)at(8,2)[color=gray,transition] {$c$};
     \node(tr'3)at(8,-1)[color=gray,transition] {$b$};

     \path(tr'2) edge [->] (s'1);
     \path(s'1) edge [->] (tr'5);
     \path(tr'5) edge [->] (s'2);
     \path(s'2) edge [->] (tr'6);
     \path(tr'6) edge [->] (s'5);
     \path(s'5) edge [->] (tr'2);
     \path(s'1) edge [bend left] [->] (tr'1);
     \path(tr'1) edge [bend left] [->] (s'5);
     \path(s'5) edge [bend right] [->] (tr'7);
     \path(tr'7) edge [->] [bend right] (s'2);

     \path(tr'4) edge [->] (s'1);
     \path(s'3) edge [->] (tr'4);
     \path(tr'3) edge [->] (s'3);
     \path(s'5) edge [->] (tr'3);


    \node[state,draw=none] (f1)at(6.25,3.5) {};
    \node[state,draw=none] (f2)at(8.5,4.5) {};
    \node[state,draw=none] (f3)at(8.5,2.5) {};
    \draw [->] (f1) to node[auto] {} (f2);
    \draw [->] (f1) to node[auto] {} (f3);

     \end{tikzpicture}
     
    \caption{Two possible ways to synchronize two agents on sequential components. The blue and red graphs represent the agents, the two gray nets (upper and lower) represent two of their possible compositions.}
    \label{fig:sync_seq}
\end{figure}
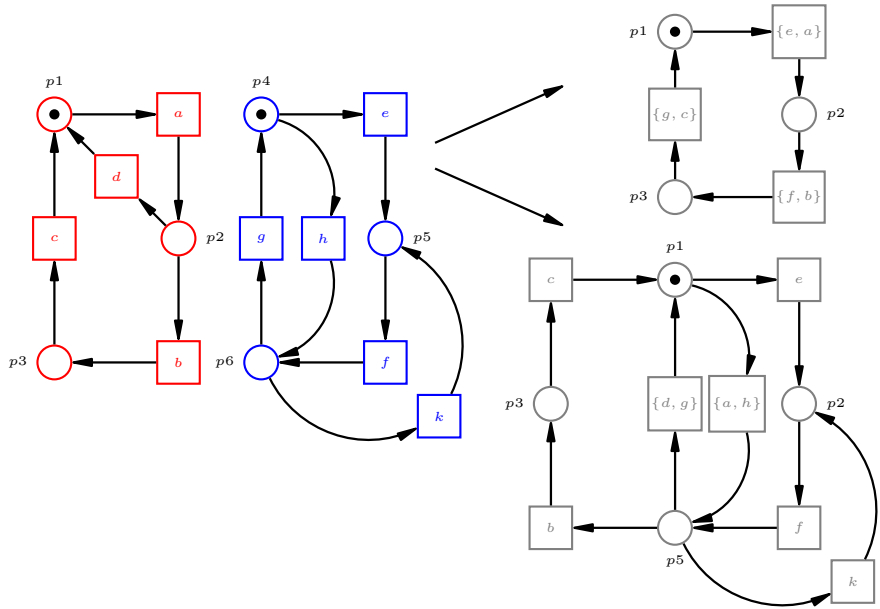

\begin{example}
\label{ex:autofusion}
    Note that fusion of sequential components may be used to merge two equivalent parts of a single system. As a representative example, one can consider the net presented in Fig.~\ref{fig:autofusion} (on the left). Its two green sequential components are merged into one using an external sequential component, the magenta net on the right. 
    For the sequential composition, we use $X_1 = X_2 = \emptyset$
    and $X'_1 = \{\{x1,\lnot x1\},\{y1,\lnot y1\}\}$, 
    $X'_2 = \{\{r1,\lnot r1\},\{r1,\lnot r1\}\}$ together with $\rho'_1(x1)=\rho'_2(y1)=r1$ and $\rho'_1(\lnot x1)=\rho'_2(\lnot y1)=\lnot r1$. The resulting net is shown in Fig.~\ref{fig:tr2glpnb}.

\begin{figure}
    \centering
        \begin{tikzpicture}[node distance=1.3cm,>=arrow30,line  width=0.3mm,scale=1.1,auto,bend angle=45,font=\tiny]
     \tikzstyle{place}=[draw,circle,thick,minimum size=4.5mm]
     \tikzstyle{transition}=
                [draw,regular polygon,thick,
                 regular polygon sides=4,minimum size=8mm, inner sep = -2pt]

     \node(w1)at(2,4)[color=red,place,tokens=1,label=above:$w1$]{};
     \node(nw1)at(1.7,3.3)[color=red,place,tokens=0,label=below:$\lnot w1$]{};
     \node(a1)at(1,1)[color=red,place,tokens=0,label=left:$a1$]{};
     \node(na1)at(1.5,1.5)[color=red,place,tokens=1,label=right:$\lnot a1$]{};
     \node(t1)at(2,2)[color=red,place,tokens=0,label=right:$t1$]{};
     \node(nt1)at(3.6,2)[color=red,place,tokens=1,label=left:$\lnot t1$]{};
     \node(n1)at(2.8,3)[color=red,transition] {$n1$};
     \node(n2)at(2.8,1)[color=red,transition] {$n2$};
     \node(n3)at(1,4)[color=red,transition] {$n3$};
     \path(n3) edge [->] (w1);
     \path(nw1) edge [->] (n3);
     \path(w1) edge [->] (n1);
     \path(n1) edge [->] (nw1);
     \path(n1) edge [->] (t1);
     \path(t1) edge [->] (n2);
     \path(n2) edge [->] (a1);
     \path(na1) edge [->] (n2);
     \path(a1) edge [->] (n3);
     \path(n3) edge [->] (na1);
     \path(nt1) edge [->] (n1);
     \path(n2) edge [->] (nt1);

     \node(r1)at(8,2.5)[color=magenta,place,tokens=0,label=right:$r1$]{};
     \node(nr1)at(9,2.5)[color=magenta,place,tokens=1,label=right:$\lnot r1$]{};
     \node(r1nr1)at(8.5,3.5)[color=magenta,transition] {$\emptyset$};
     \node(nr1r1)at(8.5,1.5)[color=magenta,transition] {$\emptyset$};

     \node(x1)at(4,1.5)[color=green,place,tokens=0,label=right:$x1$]{};
     \node(xr1)at(6,1.5)[color=green,place,tokens=1,label=right:$\lnot x1$]{};
     \node(x1nr1)at(5,2)[color=green,transition] {$\emptyset$};
     \node(nr1x1)at(5,1)[color=green,transition] {$\emptyset$};

     \node(y1)at(4,3.5)[color=green,place,tokens=0,label=right:$y1$]{};
     \node(yr1)at(6,3.5)[color=green,place,tokens=1,label=right:$\lnot y1$]{};
     \node(y1nr1)at(5,4)[color=green,transition] {$\emptyset$};
     \node(nr1y1)at(5,3)[color=green,transition] {$\emptyset$};

     \path(r1) edge [->] (r1nr1);
     \path(r1nr1) edge [->] (nr1);
     \path(nr1) edge [->] (nr1r1);
     \path(nr1r1) edge [->] (r1);

     \path(x1) edge [->] (x1nr1);
     \path(x1nr1) edge [->] (xr1);
     \path(xr1) edge [->] (nr1x1);
     \path(nr1x1) edge [->] (x1);

     \path(y1) edge [->] (y1nr1);
     \path(y1nr1) edge [->] (yr1);
     \path(yr1) edge [->] (nr1y1);
     \path(nr1y1) edge [->] (y1);

     \path [->,>=read] (y1) edge (n1);  
     \path [->,>=read] (x1) edge (n2);  
     \end{tikzpicture}
    \caption{Two green sequential components of the net on the left merged using the external, magenta net, on the right.}
    \label{fig:autofusion}
\end{figure}
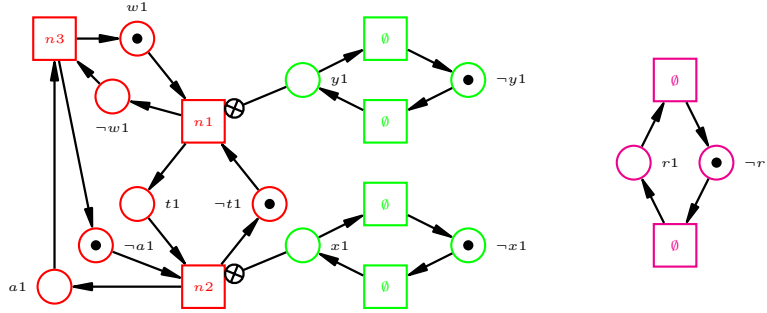
\end{example}

\begin{definition}
Let $\Sigma=(P,T,\flow,m_0)$ be a net obtained by a fusion of subnets $\Sigma_i=(P_i,T_i,\flow_i,m_{0,i})$. Then the projection of a marking $M\subseteq P$ of system $\Sigma$ to its subsystem $\Sigma_i$ is a marking $M|_{\Sigma_i}\subseteq P_i$ defined in the following way: $M|_{\Sigma_i}(p) = M(p)$ for $p\in P_i\cap P$ or $M|_{\Sigma_i}(p) = M(\rho(p))$ otherwise (in the case of sequential component fusion, when $p$ changed the name after the fusion). 
\end{definition}

\subsection{Synthesis of Petri nets from transition systems}
\label{sec:synt}
Given a labeled transition system $\TS = (\Ls, \Evt, \TR, \iota)$, 
solving the synthesis problem means finding a
Petri net $\Sigma$ such that $\MG(\Sigma)$ is isomorphic to $\TS$.
The classical techniques for the synthesis of 1-safe net systems are based on 
the research of \emph{regions} \cite{BBD15}.
Intuitively, a region is a subset of $\Ls$, and it represents a place of the net.
Each label $e\in \Evt$ is associated to a transition in the synthesized net,
and its relations with places are determined by the relations between the 
arcs labeled with $e$ in $\TS$ and the regions in $\TS$.

Formally, a \emph{region} is a subset of states $r\subseteq \Ls_i$ such that, 
for each $e\in \Evt_i$, one of the following conditions holds.
\begin{enumerate}
    \item $e$ enters the region, i.e. for each arc labeled with $e$ from 
    $s_1$ to $s_2$, $s_1 \not \in r \wedge s_2 \in r$;
    \item $e$ leaves the region, i.e. for each arc labeled with $e$ from 
    $s_1$ to $s_2$, $s_1 \in r \wedge s_2 \not\in r$; 
    \item $e$ does not cross the border of the region, i.e. for each arc 
    labeled with $e$ from $s_1$ to $s_2$, $s_1 \in r \wedge s_2 \in r$, or 
    $s_1 \not\in r \wedge s_2 \not\in r$. 
\end{enumerate}
If $e$ does not cross the border of $r$, and for each arc labeled with $e$ from 
$s_1$ to $s_2$, we have $s_1, s_2\in r$, then we can say that $e$ is \emph{inside} $r$.
If the synthesis problem has a solution, we can determine the flow relation 
in the net as follows. 
For each transition $e$, for each place $r$, $r$ is a precondition of $e$ if $e$ leaves $r$ in 
$\TS$, $r$ is a postcondition of $e$ if $e$ enters $r$ in $\TS$. 
If $e$ is inside $r$, then we can see $r$ as both a precondition and a 
postcondition of $e$, and add a self loop to the net. 
Otherwise, there is no flow relation between $r$ and $e$ in the net. 

Not every transition system can be synthesized into a 1-safe net with the procedure described above. 
In particular, we can synthesize a 1-safe net from a transition system if, 
and only if, the set of regions of the transition system satisfies the so 
called \emph{state separation property} (SSP) and \emph{event-state separation property} (ESSP): 
\begin{itemize}
   \item $\forall s, s'\in L_i \quad s\neq s'  \quad \rightarrow \exists r\in R : (s\in r \wedge s' \not\in r) \lor (s\not\in r \wedge s\in r)$\quad (SSP)
   \item $\forall e \in E, \, \forall s\in S : e$  is not outgoing from  $s$, 
    $\rightarrow \exists r : s\not\in r \wedge (e$  leaves  $r \lor e$  is inside  $r)$ \quad(ESSP)
\end{itemize}
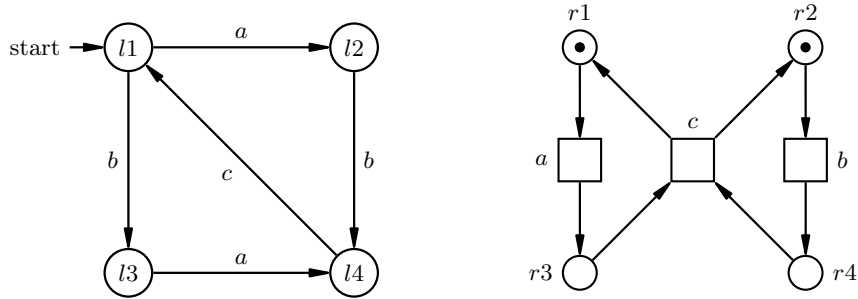
\begin{figure}
    \centering
       \begin{tikzpicture}[->,node distance=3cm,>=arrow30,line  width=0.3mm,scale=1.1,auto,bend angle=45,font=\small]
     \tikzstyle{place}=[draw,circle,thick,minimum size=4.5mm]
     \tikzstyle{transition}=
                [draw,regular polygon,thick,
                 regular polygon sides=4,minimum size=8mm, inner sep = -2pt]
                 
    \node[initial,circle,draw=black,inner sep=1pt,minimum size=18pt](A1){$l1$};
    \node[circle,draw=black,inner sep=1pt,minimum size=18pt,right of=A1] (A2){$l2$};
    \node[circle,draw=black,inner sep=1pt,minimum size=18pt,below of=A1] (A3) {$l3$};
    \node[circle,draw=black,inner sep=1pt,minimum size=18pt,below of=A2] (A4) {$l4$};

    \path (A1) edge node {$a$} (A2)
      (A3) edge node {$a$} (A4)
      (A1) edge node [xshift=-4mm] {$b$} (A3)
      (A2) edge node {$b$} (A4)
      (A4) edge node {$c$} (A1)
      ;

     \node[place,draw=black,tokens=1,label=above:$r1$,right of=A2](p1){};
     \node[place,draw=black,tokens=1,label=above:$r2$,right of= p1](p2){};
    \node[place,draw=black,tokens=0,label=left:$r3$,below of= p1](p3){};
    \node[place,draw=black,tokens=0,label=right:$r4$,below of= p2](p4){};
    \node[transition,draw=black,label=left:$a$,below of= p1,yshift=1.5cm](t1){};
    \node[transition,draw=black,label=right:$b$,below of= p2,yshift=1.5cm](t2){};
    \node[transition,draw=black,label=above:$c$,right of= t1,xshift=-1.5cm](t3){};

    \path (p1) edge [->] (t1);
    \path (t1) edge [->] (p3);
    \path (p2) edge [->] (t2);
    \path (t2) edge [->] (p4);
    \path (p3) edge [->] (t3);
    \path (p4) edge [->] (t3);
    \path (t3) edge [->] (p1);
    \path (t3) edge [->] (p2);
    ;
     \end{tikzpicture}
    \vspace*{-1cm}
    \caption{Synthesis of a 1-safe net}
    \label{fig:regions}
\end{figure}
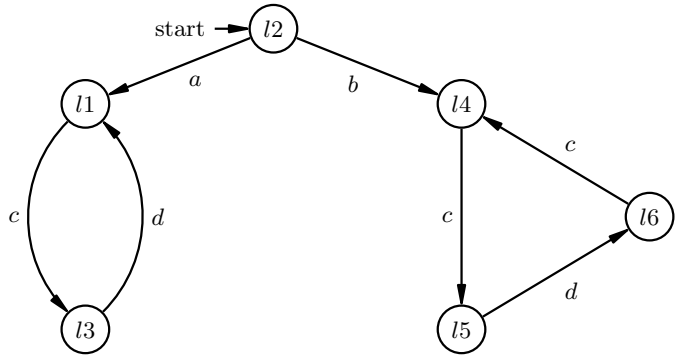
\begin{figure}
    \centering
        \begin{tikzpicture}[->,node distance=3cm,>=arrow30,line  width=0.3mm,scale=1.1,auto,bend angle=45,font=\small]
     \tikzstyle{place}=[draw,circle,thick,minimum size=4.5mm]
     \tikzstyle{transition}=
                [draw,regular polygon,thick,
                 regular polygon sides=4,minimum size=8mm, inner sep = -2pt]
                 
    \node[initial,circle,draw=black,inner sep=1pt,minimum size=18pt] (A2){$l2$};
    \node[circle,draw=black,inner sep=1pt,minimum size=18pt,left of=A2,yshift=-1cm,xshift=0.5cm](A1){$l1$};
    \node[circle,draw=black,inner sep=1pt,minimum size=18pt,right of=A2,yshift=-1cm,xshift=-0.5cm] (A4){$l4$};
    \node[circle,draw=black,inner sep=1pt,minimum size=18pt,below of=A1] (A3) {$l3$};
    \node[circle,draw=black,inner sep=1pt,minimum size=18pt,below of=A4] (A5) {$l5$};
    \node[circle,draw=black,inner sep=1pt,minimum size=18pt,right of=A4,yshift=-1.5cm,xshift=-0.5cm] (A6) {$l6$};

    \path (A1) edge [bend right] node [swap] {$c$} (A3)
      (A3) edge [bend right] node [swap] {$d$} (A1)
      (A2) edge node {$a$} (A1)
      (A2) edge node [swap] {$b$} (A4)
      (A4) edge node [swap] {$c$} (A5)
      (A5) edge node [swap] {$d$} (A6)
      (A6) edge node [swap] {$c$} (A4)
      ;

     \end{tikzpicture}
    \caption{Transition system that cannot be synthesized without a labeled net.}
    \label{fig:non-sint}
\end{figure}
\begin{example}
Consider the transition system on the left in Fig.~\ref{fig:regions}.
A $1$-safe net synthesizing it is on the right of the same figure.
Place $r1$ in the net is the region $\{l1, l3\}$ in the transition system;
since all the arcs labeled with $a$ leave the region, $r1$ is a precondition
of $a$; the arc labeled with $c$ enters the region, therefore $r1$ is a
postcondition of $c$; finally the occurrences of $b$ are either inside or
outside the region, therefore $r1$ and $b$ have no relation with each other.
Analogously for the regions $\{l1, l2\}$, $\{l2, l4\}$, $\{l3,l4\}$.
This set of regions solves all the separation properties for the
transition system.
\end{example}
However, it is always possible to obtain a labeled 1-safe system, by allowing the net to  
have more transitions with the same (singleton) label. 
In this case, we can split the transitions of $\TS$ with the same label into subgroups, and 
look for regions as if each group had a different label. 
This generates a set of different transitions in the net sharing the same label. 
To obtain such a net is always possible, since in the worse case we could consider subsets formed by 
single arcs: if each arc of the transition system is considered as if it had different 
labels from the others, it is easy to verify that the set of regions satisfies the 
separation properties SSP and ESSP. 
In this case, each state of the transition system $\TS$ is a region and therefore 
can be translated into a place of the synthesized net. 
As showed in \cite{bernardinello93}, the minimal 
regions with respect to inclusion are sufficient for the synthesis, therefore  we can consider the states of $\TS$ as all and only the places of the net. 

\begin{example}
    Let us consider the net depicted on Figure~\ref{fig:non-sint}. 
    Note that we cannot synthesize this system as a 1-safe net, because state $l_4$ need to be in all pre-regions and post-regions of $c$ at the same time.
    However, we can construct a net with three copies of $c$ and only one $d$
    (see Figure~\ref{fig:net-non-sint}).
\end{example}
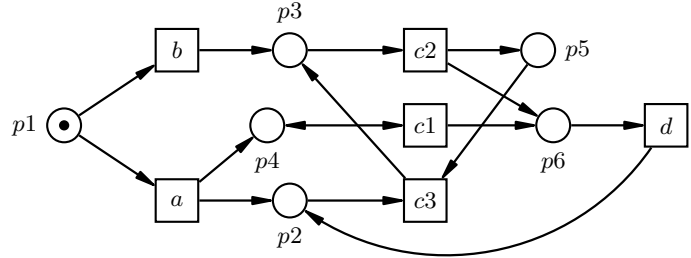
\begin{figure}
    \centering
        \begin{tikzpicture}[->,node distance=1.5cm,>=arrow30,line  width=0.3mm,scale=1.1,auto,bend angle=45,font=\small]
     \tikzstyle{place}=[draw,circle,thick,minimum size=4.5mm]
     \tikzstyle{transition}=
                [draw,regular polygon,thick,
                 regular polygon sides=4,minimum size=8mm, inner sep = -2pt]

    \node[place,draw=black,tokens=1,label=left:$p1$,right of=A2](p1){};
    
    
    \node[transition,draw=black, right of= p1, yshift=-1cm](t1){$a$};
    \node[transition,draw=black, right of= p1, yshift=1cm](t2){$b$};
    \node[place,draw=black,label=below:$p2$,right of= t1](p2){};    
    \node[place,draw=black,label=above:$p3$,right of= t2](p3){};
    \node[place,draw=black,tokens=0,label=below:$p4$,right of= t1, yshift=1cm,xshift=-0.3cm](p4){};
    
    \node[transition,draw=black,right of= p2,xshift=0.3cm](t3){$c3$};
    \node[transition,draw=black,right of= p2,yshift=1cm,xshift=0.3cm](t4){$c1$};
    
    \node[transition,draw=black,right of= p3,xshift=0.3cm](t5){$c2$};

    \node[place,draw=black,label=right:$p5$,right of= t5](p5){};
    \node[place,draw=black,label=below:$p6$,right of= t5, yshift=-1cm,xshift=0.2cm](p6){};
    
    \node[transition,draw=black,right of= p6](t6){$d$};

    \path (p1) edge [->] (t1);
    \path (p1) edge [->] (t2);
    \path (t1) edge [->] (p2);
    \path (t2) edge [->] (p3);
    \path (t1) edge [->] (p4);
    \path (p3) edge [->] (t5);
    \path (p4) edge [<->] (t4);
    \path (p2) edge [->] (t3);
    \path (t4) edge [->] (p6);
    \path (t5) edge [->] (p6);
    \path (t5) edge [->] (p5);
    \path (p6) edge [->] (t6);
    \path (t3) edge [->] (p3);
    \path (p5) edge [->] (t3);
    \path (t6) [bend left] edge [->] (p2);
    
     \end{tikzpicture}
    \caption{Net synthesized from the transition system in Fig. \ref{fig:non-sint}, obtained by splitting label $c$ in three different transitions.}
    \label{fig:net-non-sint}
\end{figure}
Let $\TS$ be a labeled transition system.
Any set of regions in $\TS$ satisfying all the separation properties 
is sufficient to construct a $1$-safe net system,
whose marking graph is isomorphic to $\TS$, independently of the chosen regions.
In particular, if a set of regions $R$ satisfies all the separation properties, any set of regions $R'$ 
such that $R\subseteq R'$ satisfies the separation properties.
Hence, we can add any place associated to a region to the synthesized net without changing its marking graph.
This property can be used for example to obtain sequential components in a $1$-safe system.

Let $\Sigma$ be a $1$-safe system, and $\MG(\Sigma)$ be its marking graph.
Let $R'$ be a set of regions partitioning the states in $\MG(\Sigma)$; if the elements of $R'$
are not yet in $\Sigma$, we can add them to the set of places.
For the properties of regions, the subsystem $\Sigma'$ of $\Sigma$ generated by $R'$ is a sequential component.
A common use of this properties consists in adding all the complementary places of those in the net.
It is easy to see that if a set of states is a region,
also its complement in the marking graph is,
and therefore, adding complementary places is always possible.
Since a region and its complement form a partition of the states of the transition system, the subnet 
that they generate forms a sequential component in the system. 
Hence, complementary places can be added to decompose a net system into sequential components.

\subsection{Synthesis of multi-agent systems}
\label{sec:synt-mas}
In this section we show how to apply the synthesis with region theory to models of multi-agent systems.
Our aim is two-fold. On the one hand, the representation of agents with Petri 
nets is at most large as the representation with transition systems, but, if
the same label appears multiple time in the agent, the Petri net may be smaller.
This is particularly important at the composition stage as will be
explained in detail in Sec.~\ref{sec:comp-pn}: even in the worse
case, the composed Petri net can be exponentially smaller than the
composition of transition systems; in the best case in which each
agent can be synthesized with region theory, the composed Petri net has only
one transition for each label of the multi-agent system.
On the other hand, the use of Petri nets could be a common framework
to unify the two semantics proposed in the previous section.

We have already noted in Sec.~\ref{sec:mas-tr} that in the multi-agent systems, each agent is a labeled transition system, therefore
we can use region theory to synthesize a 1-safe system for each of them.

From the semantics described in Sec.~\ref{sec:mas-pl}, we derive a transition system as follows.
Let $M_j = (X_j, I_j, \Ls_j, \Trext_j, \lambda_j, \iota_j, \ell, \Lbl)$ be a module.
$(p,v',q),(p',v,q')\in\Trext$.
We transform $M_j$ into the transition system $\Ag_j=(\Ls_j,\Evt_j,\TR_j,\iota_j)$, where $\Evt_j = \ell(\Trext)$, and
$\TR_j \subseteq \Ls_j\times \Evt_j\to \Ls_j$ such that 
$\TR_j(p,e) = q$ if there exists any $v\in D^{I_j}$ and $\Trext_j(p,v)=q$.
Note that in the $\Ag_j$ constructed in this way,
we lost the information about the synchronization between agents, whereas in the model derived from $\AMAS$, synchronizations can be derived from the labels of transitions.
This issue will be tackled in Sec. \ref{sec:comp-pn},
where we will show how to incorporate this information in the agents.

Since each agent in both the semantics can be represented as a labeled transition system, we
can apply to each agent $\Ag_j = (\Ls_j, \Evt_j, \TR_j, \iota_j)$ 
the synthesis procedure as described in the previous section, and obtain a $1$-safe system agent defined
as $\Sigma_j = (R_j, T_j, \flow_j, m_{o,j}, \ell_j)$,
such that:
\begin{itemize}
    \item $R_j$ is a set of regions in $\Ag_j$ solving all the separation problems; 
    \item 
$T_j$ is the set of transitions in the $1$-safe system, each of them is associated with one or more elements in $\TR_j$; 
    \item 
$\flow_j$ is the flow relation, fully determined by $\TR_j$ as explained in the previous section; 
\item $m_{0,j}$ is the initial 
marking, that coincides with $\iota_j$; 
    \item 
$\ell_j: T_j \rightarrow 2^{\Evt_j}$ is the  labeling function, associating 
every transition of the net with the label of the corresponding arc on the agent (formally as a singleton). 
\end{itemize}
\begin{example} 
\label{ex:TGCsem1}
Consider the $\AMAS$ in Figure~\ref{fig:tramas}.
In Figure~\ref{fig:tr2pn}, each agent has been synthesized into a Petri net. 
In this case, each agent is trivially synthesizable in an exact 
way: since each action is never repeated inside the same agent, 
each of their states is a minimal region, and the set of minimal
regions satisfy all the separation properties.
\begin{figure}
    \centering
        \begin{tikzpicture}[node distance=1.3cm,>=arrow30,line  width=0.3mm,scale=1.1,auto,bend angle=45,font=\tiny]
     \tikzstyle{place}=[draw,circle,thick,minimum size=4.5mm]
     \tikzstyle{transition}=
                [draw,regular polygon,thick,
                 regular polygon sides=4,minimum size=8mm, inner sep = -2pt]

     \node(w1)at(2,4)[color=red,place,tokens=1,label=above:$w1$]{};
     \node(a1)at(1,1)[color=red,place,tokens=0,label=left:$a1$]{};
     \node(t1)at(2,2)[color=red,place,tokens=0,label=right:$t1$]{};
     \node(n1)at(2.8,3)[color=red,transition] {$n1$};
     \node(n2)at(2.8,1)[color=red,transition] {$n2$};
     \node(n3)at(1,4)[color=red,transition] {$n3$};
     \path(n3) edge [->] (w1);
     \path(w1) edge [->] (n1);
     \path(n1) edge [->] (t1);
     \path(t1) edge [->] (n2);
     \path(n2) edge [->] (a1);
     \path(a1) edge [->] (n3);

     \node(g)at(5,4)[color=green,place,tokens=1,label=above:$g$]{};
     \node(r1)at(4.5,2)[color=green,place,tokens=0,label=left:$r1$]{};
     \node(r2)at(5.5,2)[color=green,place,tokens=0,label=right:$r2$]{};
     \node(sn1)at(4,3)[color=green,transition] {$n1$};
     \node(sn2)at(4,1)[color=green,transition] {$n2$};
     \node(sm1)at(6,3)[color=green,transition] {$m1$};
     \node(sm2)at(6,1)[color=green,transition] {$m2$};
     \path(g) edge [->] (sn1);
     \path(sn1) edge [->] (r1);
     \path(r1) edge [->] (sn2);
     \path [->] (sn2) edge [bend right,out=-60,in=-180,min distance=1cm] (g);
     \path(g) edge [->] (sm1);
     \path(sm1) edge [->] (r2);
     \path(r2) edge [->] (sm2);
     \path [->] (sm2) edge [bend left,out=60,in=180,min distance=1cm] (g);

          \node(w2)at(8,4)[color=blue,place,tokens=1,label=above:$w2$]{};
     \node(a2)at(9,1)[color=blue,place,tokens=0,label=right:$a2$]{};
     \node(t2)at(8,2)[color=blue,place,tokens=0,label=left:$t2$]{};
     \node(m1)at(7.2,3)[color=blue,transition] {$m1$};
     \node(m2)at(7.2,1)[color=blue,transition] {$m2$};
     \node(m3)at(9,4)[color=blue,transition] {$m3$};
     \path(m3) edge [->] (w2);
     \path(w2) edge [->] (m1);
     \path(m1) edge [->] (t2);
     \path(t2) edge [->] (m2);
     \path(m2) edge [->] (a2);
     \path(a2) edge [->] (m3);
     
     \end{tikzpicture}
    \caption{1-safe Petri nets for transition systems depicted in Fig.~\ref{fig:tramas}.}
    \label{fig:tr2pn}
\end{figure}
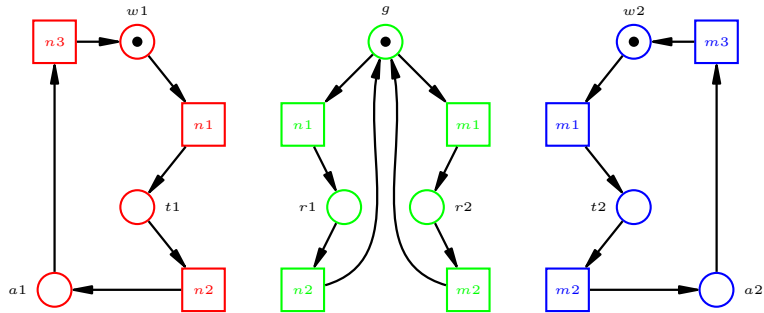

\end{example}
\begin{remark}
In the synthesis of agents, we do not need
read or inhibitor arcs,
as they can be simulated by classical arcs: If transition $t$ can only occur when place $p$ is marked, we can add two relations to F, namely $(p, t)$ and $(t, p)$. In this way, $p$ become a precondition of $t$, therefore necessary for its occurrence, and it assumes the same value after its occurrence. The inhibitor arcs can be simulated analogously by adding the complementary place of $p$. The marking graph obtained with this construction is isomorphic to the one obtained with read and inhibitor arcs. However, the semantics obtained through the two constructions are not equivalent: With only classical arcs, the agent would consume the token in place $p$, instead of just observing its value, as would happen in the case of read and inhibitor arcs. 
When we consider individual agents, we allow for the classical arcs semantics, whereas for interfaces in the data-based model we find the use of read arcs more precise. 
\end{remark}
\begin{remark}
\label{rem:non-unique}
In general, we can have different labeled 1-safe systems that are able to achieve the same transition system. Among the possible solutions, we prefer those that have fewer places. Especially, we want to avoid complementary places if they are not necessary. In this way, we increase the chance to obtain non-trivial sequential components. Another important issue is related to the situation where the considered system cannot be synthesized as a 1-safe Petri net. A possible solution is label-splitting, i.e. choosing transitions with the same label that violates a separation property and divide such a set into two parts in such a way that the considered separation property is satisfied. Naturally, it is good to minimize the number of transitions in the result, and hence also minimize the number of splittings.
\end{remark}

\section{Composition with Petri nets in the two semantics}
\label{sec:comp-pn}
In this section we show how to compose Petri net agents in two different ways, in order to obtain global systems with an equivalent behavior to those described in Sec. \ref{sec:mas-ts}. In particular, Sec. \ref{sec:tr} focuses on the synchronization on common transitions, analogously to Sec. \ref{sec:mas-tr} on transition systems, whereas Sec. \ref{sec:pn-pl} proposes the synchronization on sequential components, with the aim to produce a global system analogous to the one defined in Sec. \ref{sec:mas-pl}.
\subsection{Synchronization on common actions}
\label{sec:tr}
In this section we provide the construction of a \emph{global} net showing the 
interaction of the agents when they synchronize on transitions:

\begin{definition}
\label{def:synch_act}
Let $\Sigma_1 = (P_1, T_1, F_1, \iota_1, \ell_1), ..., \Sigma_n = (P_n, T_n, F_n, \iota_n, \ell_n)$ be the set of Petri net agents.
The action based composition of those nets is a net $\Sigma = (P, T, F, m_0, \ell)$ such that:
\begin{itemize}
    \item $P$ is the union of the sets of places $P_i$.
    \item $T = \bigcup_{\alpha\in \bigcup_{i=1,\ldots,n}\ell_i(T_i)}\bigotimes_{i \in \{1,...,n\}}T_i^\alpha$, where $T_i^\alpha = \{t\in T_i : \ell_i(t) = \alpha\}$ is the set of transitions used in the net $\Sigma_i$ and labeled with $\alpha$.
    \item The flow relation is determined as follows: for each transition $t \in T$, and each place $p\in P$ there is an arc from $p$ to $t$ if there is a $\Sigma_i$ and $t_j\in T_i$ such that $p\in P_i$, $t_j$ is a component in $t$, and $(p, t_j) \in F_i$; analogously for the arcs from $t\in T$ to $p\in P$. 
    \item The initial marking $m_0$ is the union of all the elements $\iota_i$, with  $i\in \{1,\ldots, n\}$.
    \item The labeling function $\ell$ associates every transition $t\in T$ to the  union of labels of all its components, that is a singleton by construction (as the union of positive number of equal singletons and empty sets).
\end{itemize}
We denote the alphabet $\Lambda$ of $\Sigma$ with $\ell(T)$.
\end{definition}

By the construction presented as Definition~\ref{def:synch_act}, each place in $\Sigma$ belongs at most to one agent, 
whereas the transitions can be shared. 
Note that some of the transitions may be enabled in no reachable marking, 
and therefore are not live. 
As we discussed in Sec. \ref{sec:tr}, and more in detail in \cite{AM23}, the same problem happens when we consider the composition of $\AMAS$, because some 
states may not be reachable from the initial state $\iota$, due to the synchronization 
constraints. 
\begin{figure}
    \centering
        \begin{tikzpicture}[node distance=1.3cm,>=arrow30,line  width=0.3mm,scale=1.1,auto,bend angle=45,font=\tiny]
     \tikzstyle{place}=[draw,circle,thick,minimum size=4.5mm]
     \tikzstyle{transition}=
                [draw,regular polygon,thick,
                 regular polygon sides=4,minimum size=8mm, inner sep = -2pt]

     \node(w1)at(3,4)[color=red,place,tokens=1,label=above:$w1$]{};
     \node(a1)at(2,1)[color=red,place,tokens=0,label=left:$a1$]{};
     \node(t1)at(3,2)[color=red,place,tokens=0,label=right:$t1$]{};
     \node(n1)at(3.8,3)[color=red,transition] {$n1$};
     \node(n2)at(3.8,1)[color=red,transition] {$n2$};
     \node(n3)at(2,4)[color=red,transition] {$n3$};
     \path(n3) edge [->] (w1);
     \path(w1) edge [->] (n1);
     \path(n1) edge [->] (t1);
     \path(t1) edge [->] (n2);
     \path(n2) edge [->] (a1);
     \path(a1) edge [->] (n3);

     \node(g)at(5,4)[color=green,place,tokens=1,label=above:$g$]{};
     \node(r1)at(4.5,2)[color=green,place,tokens=0,label=left:$r1$]{};
     \node(r2)at(5.5,2)[color=green,place,tokens=0,label=right:$r2$]{};
     \node(sn1)at(4,3)[color=green,transition] {$n1$};
     \node(sn2)at(4,1)[color=green,transition] {$n2$};
     \node(sm1)at(6,3)[color=green,transition] {$m1$};
     \node(sm2)at(6,1)[color=green,transition] {$m2$};
     \path(g) edge [->] (sn1);
     \path(sn1) edge [->] (r1);
     \path(r1) edge [->] (sn2);
     \path [->] (sn2) edge [bend right,out=-60,in=-180,min distance=1cm] (g);
     \path(g) edge [->] (sm1);
     \path(sm1) edge [->] (r2);
     \path(r2) edge [->] (sm2);
     \path [->] (sm2) edge [bend left,out=60,in=180,min distance=1cm] (g);

          \node(w2)at(7,4)[color=blue,place,tokens=1,label=above:$w2$]{};
     \node(a2)at(8,1)[color=blue,place,tokens=0,label=right:$a2$]{};
     \node(t2)at(7,2)[color=blue,place,tokens=0,label=left:$t2$]{};
     \node(m1)at(6.2,3)[color=blue,transition] {$m1$};
     \node(m2)at(6.2,1)[color=blue,transition] {$m2$};
     \node(m3)at(8,4)[color=blue,transition] {$m3$};
     \path(m3) edge [->] (w2);
     \path(w2) edge [->] (m1);
     \path(m1) edge [->] (t2);
     \path(t2) edge [->] (m2);
     \path(m2) edge [->] (a2);
     \path(a2) edge [->] (m3);
     
     \end{tikzpicture}
    \caption{Global Petri net model resulting from the composition of the nets depicted in Fig.~\ref{fig:tr2pn}.}
    \label{fig:tr2glpna}
\end{figure}
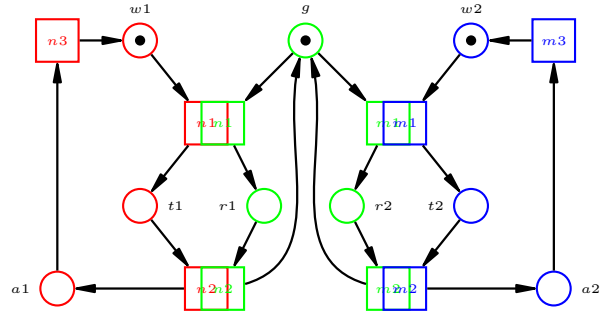

\begin{example}
\label{ex:small_parts}
Fig.~\ref{fig:tr2glpna} represents the composition
of the three agents from Fig. \ref{fig:tr2pn} based on synchronizations in transitions. 
Since in this case every agent is synthesizable, in the global 
net in Fig.~\ref{fig:tr2glpna} each transition has 
a different label. 

This is in general not the case. Consider, as an example, the Petri net in Fig.~\ref{fig:net-non-sint}, where three transitions share the label $c$. If we need to synchronize this agent with another sharing label $c$, we need to fuse all the occurrences of $c$ in the first agent, hence $c1$, $c2$ and $c3$, with all the occurrences of $c$ in the second agent.
Hence, let $n_1$ be the number of occurrences of $c$ in the first agent and $n_2$ be the number of occurrences of $c$ in the second agent, the result of the composition of the two Petri nets will have $n_1 \times n_2$ transitions labeled as $c$.  
\end{example}
Let $I_r$ be the canonical IIS where all the unreachable states and transitions have been removed. 
The following proposition shows that
synthesis and composition are commutative, 
i.e synthesizing Petri net agents from the $\AMAS$ and then composing them is equivalent to 
construct the composition of $\AMAS$ and then synthesizing a Petri net.
As already noted in Remark \ref{rem:non-unique}, the synthesis of labeled 1-safe
systems, may produce different sets of Petri net agents. However, for all of
them the commutativity holds.
\begin{proposition}
Let $S = (\Ag_1,...,\Ag_n)$ be an $\AMAS$,
$\Sigma_1,...,\Sigma_n$ be a set of Petri
net agents such that for each $i\in\{1,...,n\}$, $\Sigma_i$ is a 1-safe system, possibly labeled, obtained through the synthesis of $\Ag_i$,
$\Sigma = (P, T, F, m_0, \ell)$ be global Petri net 
constructed as described above, and $I$ the canonical IIS of $S$. 
The marking graph of $\Sigma$ is isomorphic to $I_r$.
\end{proposition}
\begin{proof}
The proof is based on the fact that the set
of places $P$ of the global net can be partitioned into
the places of the agents $\Sigma_i$, and each of them
has a marking graph isomorphic to an agent $\Ag_i$.

We first show that the set of labels of the
transitions enabled in $m_0$ in $\Sigma$ is
the same as the set of labels of the transitions
occurring in the initial state $\iota$ of $I_r$, 
and that for each label $\alpha$, the cardinality of the set of transitions
associated to $\alpha$ is the same.

Assume that $\iota$ enables a transition with label $\alpha$.
Let $\Ag_\alpha = \{\Ag_j : j\in J\subseteq \{1,...,n\}\}$ be the set of agents in the $\AMAS$ with $\alpha$ in their alphabet.
All the $\Ag_j\in \Ag_\alpha$ must
enable a transition labeled with $\alpha$
in their initial state $\iota_j$.
By contradiction, assume that no transition
labeled with $\alpha$ is enabled in $m_0$;
since $m_0$ is the union of the initial states of the $\Sigma_i$,
there must be a $\Sigma_i$ with $\alpha$ in its alphabet
that does not enable any transition labeled with $\alpha$ in its initial state.
This is impossible, since by construction, 
$\MG(\Sigma_i)$ is isomorphic to  $\Ag_i\in \Ag_\alpha$.
Similarly, if there is a transition with a
label $\alpha$ enabled in $m_0$, then each
Petri net agent having it in its alphabet must have it enabled 
in the initial state, since by construction,
the marking graph of the Petri net agents are isomorphic
to the agents in $S$, $\alpha$ must be enabled in
the initial state also in $I_r$.

For each label $\alpha$, the number of transitions 
labeled $\alpha$ and enabled in $\iota$
is the product of the numbers of transitions
enabled in the initial states of each agent in $\Ag_\alpha$;
since the marking graphs of the Petri net agents
are isomorphic to the agents in $\Ag_\alpha$,
in $m_0$ the same number of transitions labeled with $\alpha$ is enabled.

Next, we show that two transitions bring to different markings in $\Sigma$,
iff they bring to different states in $I_r$.
Assume that two transitions $t_1$, $t_2$
enabled in $\iota$
arrive in different states of $I_r$, then by construction,
there is at least an agent $\Ag_i$ participating in the action
and such that the local state after the occurrence of $t_1$
differs from the local state after the occurrence of $t_2$.
This must happen also in $\Sigma_i$, 
since its marking graph is isomorphic to $\Ag_i$,
therefore the markings reached from $m_0$ in $\Sigma$ differs at 
least for the places belonging to $\Sigma_i$.
Analogously, if $t_1$ and $t_2$ lead from $\iota$ to the same state in $I_r$, 
then for all the agents participating in them, 
the local states after $t_1$ and $t_2$ must be the same,
and this is true also for all the $\Sigma_i$ participating in $t_1$ and $t_2$, and therefore also for their union.

This same reasoning can be applied recursively to the states reached from the initial marking,
until considering all the reachable states.
With the same argument we can also show that a cycle is closed on the marking graph of $\Sigma$
iff it is closed in $I_r$.
\end{proof}
Note that the result stated in the proposition above holds
for any 1-safe labeled system  synthesized
from a transition system describing behaviors of particular agents and their compositions.

\subsection{Synchronization based on sequential components}
\label{sec:pn-pl}
In the case of the data-driven synchronization, the synthesis of the single module is more sophisticated. 
We have already described how to obtain a net from the underlying transition system; here we show which additional information we need to incorporate in the agents. 
To improve readability, we divide the procedure into stages.

First, let us consider a module $M=(X,I,\Ls,\Trext,\lambda,\iota,\ell)$, as in the Definition~\ref{d:module}. 
We use the synthesis procedure based on the region theory and described in Section~\ref{sec:synt-mas} obtaining a 1-safe Petri net $\Sigma = (R, T, \flow, \iota, \ell)$.
Note that, at this stage, we forget about the values of the variables (both internal and external). In the following stages we need to restore them.

We start by adding internal variables of $M$ to the net $\Sigma$.
Let us consider a variable $x\in X$ and all of the possible valuations of this variable $val_x = \{v_1,\ldots,v_n\}$. 
First we note that the set of local states
$\Ls_{x_i} = \{q\in\Ls \;|\; \lambda(q)(x) = v_i\}$ 
forms a region in $\Ag$ and we can add a fresh place called $x_i$ to $\Sigma$ (modifying $F$ and $\iota$ accordingly). 
Indeed, directly by the construction, if $t1,t2\in T$ and
$\ell(t1) = \ell(t2) = a \in \Evt$ and $t1$ or $t2$ 
is enabled at $p,q\in\Ls$ then
$\lambda(p) = \lambda(q)$ (hence also $\lambda(p)(x) = \lambda(q)(x)$).
Similarly, if $t1,t2\in T$ and
$\ell(t1) = \ell(t2) = a \in \Evt$ and $t1$ or $t2$ 
lead to some $p,q\in\Ls$ then
$\lambda(p) = \lambda(q)$ (hence also $\lambda(p)(x) = \lambda(q)(x)$).
Hence every transition labeled with $a$ either enters $\Ls_{x_i}$,
or exits from $\Ls_{x_i}$ or is independent with $\Ls_{x_i}$.
Moreover, the family $\Ls_x = \{\Ls_{x_i} \;|\; i\in val_x\}$ forms 
a partition of $\Ls$ into regions. 
Hence, as discussed in Sec.~\ref{sec:synt}, $\Ls_x$ is a sequential component of
$\Sigma' = (R\cup \{x_1,\ldots,x_n\}, T, \flow', \iota', \ell\})$
and adding those places does not change the behavior of $\Sigma$
(i.e. the reachability graphs of $\Sigma$ and $\Sigma'$ are equivalent).

We repeat the procedure for all internal variables of $M$ 
obtaining $\Sigma_{var} = (P_{var}, T, \flow_{var}, \iota_{var}, \ell)$:

\begin{definition}
\label{def:netwithvars}
Let $\Sigma = (R, T, \flow, \iota, \ell)$ be a 1-safe Petri net that is an effect of synthesis of module $M=(X,I,\Ls,\Trext,\lambda,\iota_M,\ell_M)$.
We construct $\Sigma_{var} = (P_{var}, T, \flow_{var}, \iota_{var}, \ell)$, where
\begin{itemize}
    \item $P_{var} = R \cup \bigcup_{x\in X} \{x_1,\ldots,x_{n_x}\}$, where $n_x$ is a maximal value that the variable $x$ can get.
    \item $\flow_{var}$ is $\flow$ enriched in the way described above (by changes of the values related to the introduced regions).
    \item $\iota_{var}(p) = \iota(p)$ for $p\in R$, $\iota_{var}(x_{\iota_M(x)})=1$, and $\iota_{var}(p)=0$ otherwise.
\end{itemize}
\end{definition}

\begin{remark}
In some cases the places representing the value 
of internal variable may have already been added 
to the $1$-safe net during the synthesis phase 
described in Sec.~\ref{sec:synt-mas}. 
If this is the case, we can refine the net $\Sigma_{var}$ by not adding the new places,
and just noting which place represents a certain value of an internal variable.

\end{remark}

In the next stage, we will add external variables to mimic the operation of the modeled system by restricting the enabledness of the transitions from $T_{var}$.

Consider all suitable external variables $y_i\in I$ and all possible valuations of those variables $val_{y_i} = \{w_{i,1},\ldots,w_{i,m_i}\}$ and construct places $P_{ext}=\{y_{i,j} \;|\; i\in I \wedge 1\leq j\leq m_i\}$ similarly to the case of internal variables. We also allow for a free circulation of tokens
among different values of the same variable by defining
$T_{ext} = \{t_{i,j,k} \;|\; i\in I \wedge 1\leq j,k \leq m_i \wedge j\neq k\}$ and $\flow_{ext} = \{(y_{i,j},t_{i,j,k}) \;|\; y_{i,j}\in P_{ext} \wedge t_{i,j,k}\in T_{ext}\} \cup \{(t_{i,j,k},y_{i,k}) \;|\; y_{i,k}\in P_{ext} \wedge t_{i,j,k}\in T_{ext}\}$.
To leave the original names of the transitions that change the values of the internal variables, we take $\ell(t)=\emptyset$ for all transitions related to the external variables. Recalling Example~\ref{ex:autofusion}, we can define separate components for all external variables - they will be properly merged in case of multiple use of internal variables. Note that the ranges of variability must match.

We need to guarantee that transitions from $T_{var}$ are enabled only in the favorable circumstances.
Namely, only if the valuation of additional places corresponds
to one of the conditional transitions from $\Trext$.
Usually there are many solutions for this goal, but the common problem is the valuation of the external variables in the initial state of closed MAS. For this purpose we define a parameterized and straight-forward solution as follows. 

Let us consider a transition $t\in T_{var}$ and all states $\Ls_t$
at which $t$ is enabled in $\Sigma_{var}$. 
We check whether there are $i \in I$ from which $t$ is independent,
namely $t$ can occur for any value of $i$.
For each $t$, we call $I_{\lnot t}$ the set of external variables satisfying this 
property, $I_d = I \setminus I_{\lnot t}$, and $D^{I_d}$ the set of evaluations of the variables in $I_d$.
By $val_t = \{val\in D^{I_d} \;|\; \ell(q,val,p)=t\}$, we define the set of all admissible  valuations of external variables, for a given $t$.

We define a Petri net $\Sigma_{int}^{init}$ with read arcs, where $init\in D^I$. 

\begin{definition}
\label{d:modulef}
    Let $M=(X,I,\Ls,\Trext,\lambda,\iota_M,\ell_M)$ be a module,
    $init\in D^I$
    and 
    $\Sigma_{var} = (P_{var}, T, \flow_{var}, \iota_{var}, \ell)$ be a Petri net
    obtained as described in Definition~\ref{def:netwithvars}.
    We define $\Sigma^{init}_{int} = (P_{int}, T_{int}, \flow_{int}, \read, \iota_{int}, \ell_{int})$, where:
    \begin{itemize}
        \item $P_{int} = P_{var}\cup P_{ext}$.
        \item $\iota_{int} = \iota_{var}\cup\{y_{i,init(y_i)} \;|\; y_i\in I\}$.
        \item $T_{int} = T_{ext}\cup\bigcup_{t\in T_{var}}\;T_t$, where
        $T_t = \{t_{val} \;|\; val \in val_t\}$.
        \item $\flow_{int} = \flow_{ext}\cup\{(p,t_{val}) | (p,t)\in \flow_{var}\wedge val\in val_t\} \cup\{(t_{val},p) | (t,p)\in \flow_{var}\wedge val\in val_t\}$.
        \item $\ell_{int}(t_{val}) = \ell_{var}(t)$ for $t\in T_{ext}$ and $\ell_{int}(t)=\emptyset$ for $t\in T_{ext}$.
        \item $\read = \{(y_{i,j},t_{val}) \;|\; t_{val}\in T_{int} \wedge val(y_i) = j\}$.
    \end{itemize}
\end{definition}

Note that in the constructed net $\Sigma_{int}$ we can simulate all the computations of $\Sigma$ by  
tuning the external part of the net between each consecutive transitions from $T$.

For each external variable $w$, the places in
$P_{ext}$ generate a sequential component $\Sigma_w$ of $\Sigma_{int}$. For each $\Sigma_w$ we assume
\emph{super-fair randomness}, namely, let
$T_w$ be the set of transitions in $\Sigma_w$ and $t\in T_w$; $t$ needs to occur infinitely often
in $\Sigma_{int}$.
\begin{example}
The net $\Sigma_{int}$ from Fig.~\ref{fig:tr2glpnb} refers to agent $\Sigma_1$ of Fig. \ref{fig:tr2pn}, closed only with respect to the controller.
Note that $n1$ depends only on the variable $r1$, and it can occur only when for a single value of the variable (since $r1$ can take only $2$ values), therefore we do not need to split the transition further. 
Similarly, neither $n3$ or $n2$ need to be split, since they are independent from all the external variables.

\begin{figure}
    \centering
        \begin{tikzpicture}[node distance=1.3cm,>=arrow30,line  width=0.3mm,scale=1.1,auto,bend angle=45,font=\tiny]
     \tikzstyle{place}=[draw,circle,thick,minimum size=4.5mm]
     \tikzstyle{transition}=
                [draw,regular polygon,thick,
                 regular polygon sides=4,minimum size=8mm, inner sep = -2pt]

     \node(w1)at(2,4)[color=red,place,tokens=1,label=above:$w1$]{};
     \node(nw1)at(1.7,3.3)[color=red,place,tokens=0,label=below:$\lnot w1$]{};
     \node(a1)at(1,1)[color=red,place,tokens=0,label=left:$a1$]{};
     \node(na1)at(1.5,1.5)[color=red,place,tokens=1,label=right:$\lnot a1$]{};
     \node(t1)at(2,2)[color=red,place,tokens=0,label=right:$t1$]{};
     \node(nt1)at(3.6,2)[color=red,place,tokens=1,label=left:$\lnot t1$]{};
     \node(n1)at(2.8,3)[color=red,transition] {$n1$};
     \node(n2)at(2.8,1)[color=red,transition] {$n2$};
     \node(n3)at(1,4)[color=red,transition] {$n3$};
     \path(n3) edge [->] (w1);
     \path(nw1) edge [->] (n3);
     \path(w1) edge [->] (n1);
     \path(n1) edge [->] (nw1);
     \path(n1) edge [->] (t1);
     \path(t1) edge [->] (n2);
     \path(n2) edge [->] (a1);
     \path(na1) edge [->] (n2);
     \path(a1) edge [->] (n3);
     \path(n3) edge [->] (na1);
     \path(nt1) edge [->] (n1);
     \path(n2) edge [->] (nt1);

     \node(r1)at(5,2)[color=green,place,tokens=0,label=right:$r1$]{};
     \node(nr1)at(6,2)[color=green,place,tokens=1,label=right:$\lnot r1$]{};
     \node(r1nr1)at(5.5,3)[color=green,transition] {$\emptyset$};
     \node(nr1r1)at(5.5,1)[color=green,transition] {$\emptyset$};
     \path(r1) edge [->] (r1nr1);
     \path(r1nr1) edge [->] (nr1);
     \path(nr1) edge [->] (nr1r1);
     \path(nr1r1) edge [->] (r1);
     \path [->,>=read] (r1) edge [bend right] (n1);  
     \path [->,>=read] (r1) edge [bend left] (n2);  
     \end{tikzpicture}
    \caption{The net $\Sigma_{int}$.}  
    \label{fig:tr2glpnb}
\end{figure}
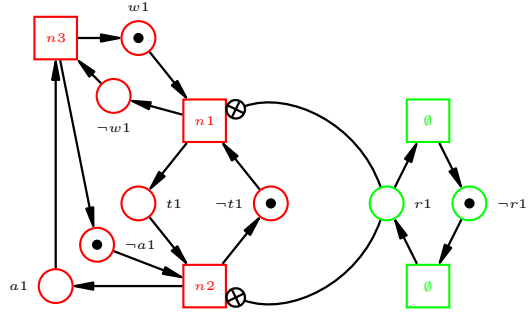

\end{example}

\begin{proposition}
\label{prop:iso-agents}
The marking graph of $\Sigma_{int}$ is isomorphic to the underlying graph of the closed module under the assumption of super-fair randomness.
\end{proposition}
\begin{proof}
By construction, the marking graph of $\Sigma_{var}$ is isomorphic to the 
underlying transition system of the open module $M$,
since it is obtained from the synthesis with regions.
We need to proof that the subsequent operations reproduce on the Petri net 
the closure operation for modules, in terms of its sequential behavior.
Let $\{((x,v),(x,w)) \;|\; v(i)\neq w(i) \wedge \forall_{j\neq i} v(j)=w(j)\} \subseteq \Trext'$
be the set of transition as in Def.~\ref{d:cl_module}; in the Petri net, this set of transitions
is reproduced by the elements in $T_{ext}$.
The initial marking of $\Sigma$ and of the closed model is isomorphic by construction,
since we can assign it in the same way for both.
Finally, the splitting of each transition $t$ into the set $T_t$ in the net and the read
arcs connecting them to places associated to external variables ensure that 
for each external evaluation allowing for the transition there is a transition
in the marking graph, and vice versa.
\end{proof}

Given a set of Petri net agents $\Sigma_{1, int},...,\Sigma_{n, int}$
derived as described above, we can construct
the global Petri net by synchronizing the system on the common sequential components.
By construction, for each agent $\Sigma_{i,int}$,
for each sequential component $\Sigma_{i,x, ext}$ referring to an external variable $x$ in $\Sigma_{i,int}$, 
there is an agent $\Sigma_{j,int}$ such that 
$x$ is an internal variable for $\Sigma_{j, int}$. Therefore, there is a sequential component
$\Sigma_{j,x,val}$ in $\Sigma_{j, val}$
such that the places in $\Sigma_{i,x,ext}$ and $\Sigma_{j,x,val}$ overlap, and we can apply the
composition through sequential components defined in Sec.~\ref{sec:pn}.
We have all the possible transitions between different values of any external variable that match with utilized transitions that change the values of the internal variable; hence
the $1$-safe system obtained by synchronizing all the external components with internal sequential components of the other modules is
the global model of the multi-agent Petri net system.

\begin{proposition}
Let $\Sigma$ be the global net with synchronizations on sequential components 
as described above, and $G$ be the closure of the global model obtained through the synchronization
of reactive modules. The marking graph of $\Sigma$ is isomorphic to the reachable part of $G$.
\end{proposition}
\begin{proof}
We show that the proposition holds for two agents.
Since the composition of $n$ agents can be seen as the composition of the first $n - 1$ agents with the $n$ agent,
the proof extends to the composition of any number of agents.

Let $M_1$ and $M_2$ be two compatible modules, and $M$ their composition, and
$\Sigma_1 = (P_1, T_1, \flow_1, \read_1, m_{0,1}, \ell_1)$ and $\Sigma_2 = (P_2, T_2, \flow_2, \read_2, m_{0,2}, \ell_2)$ 
be two Petri net agents derived from $M_1$ and $M_2$ as described above, 
with $\Sigma$ being their composition on sequential components.
The transition system $G$ and $\MG(\Sigma)$ have the same initial states by construction.
In $\Sigma$, the enabled transitions are those enabled in $m_{0,1}$ and $m_{0,2}$ 
that do not depend on external variables, those enabled in $m_{0,1}$ such that $m_{0,2}$
allows for their execution and those enabled in $m_{0,2}$ such that $m_{0,1}$ allows
for their execution.
In $G$ for each of these transitions, there must be an enabled transition with the same label, since by Prop. \ref{prop:iso-agents} the transitions are in $M_1$ and $M_2$ and they are allowed to occur;
for the same reason, no other transition can occur in the initial evaluation without being simultaneously allowed to occur in the initial state of $\Sigma$.
A consequence of Prop. \ref{prop:iso-agents} is also that the effect on the initial state of a transition $t$ occurring in $\MG(\Sigma)$ is analogous
to the effect on the initial state on $G$
of the occurrence of transition $t'$, hence $\ell_\Sigma(t) = \ell_G(t')$.
The same reasoning can be applied analogously to all reachable states, proving the proposition by structural induction.
\end{proof}
\begin{figure}
    \centering
        \begin{tikzpicture}[node distance=1.3cm,>=arrow30,line  width=0.3mm,scale=1.1,auto,bend angle=45,font=\tiny]
     \tikzstyle{place}=[draw,circle,thick,minimum size=4.5mm]
     \tikzstyle{transition}=
                [draw,regular polygon,thick,
                 regular polygon sides=4,minimum size=8mm, inner sep = -2pt]

     \node(w1)at(2,4)[color=red,place,tokens=1,label=above:$w1$]{};
     \node(nw1)at(1.7,3.3)[color=red,place,tokens=0,label=below:$\lnot w1$]{};
     \node(a1)at(1,1)[color=red,place,tokens=0,label=left:$a1$]{};
     \node(na1)at(1.5,1.5)[color=red,place,tokens=1,label=right:$\lnot a1$]{};
     \node(t1)at(2,2)[color=red,place,tokens=0,label=right:$t1$]{};
     \node(nt1)at(3.6,2)[color=red,place,tokens=1,label=left:$\lnot t1$]{};
     \node(n1)at(2.8,3)[color=red,transition] {$n1$};
     \node(n2)at(2.8,1)[color=red,transition] {$n2$};
     \node(n3)at(1,4)[color=red,transition] {$n3$};
     \path(n3) edge [->] (w1);
     \path(nw1) edge [->] (n3);
     \path(w1) edge [->] (n1);
     \path(n1) edge [->] (nw1);
     \path(n1) edge [->] (t1);
     \path(t1) edge [->] (n2);
     \path(n2) edge [->] (a1);
     \path(na1) edge [->] (n2);
     \path(a1) edge [->] (n3);
     \path(n3) edge [->] (na1);
     \path(nt1) edge [->] (n1);
     \path(n2) edge [->] (nt1);

     \node(g)at(6.5,4)[color=green,place,tokens=1,label=above:$g$]{};
     \node(ng)at(6.5,2.5)[color=green,place,tokens=0,label=above:$\lnot g$]{};
     \node(r1)at(4.6,2)[color=green,place,tokens=0,label=below:$r1$]{};
     \node(nr1)at(5.8,2)[color=green,place,tokens=1,label=left:$\lnot r1$]{};
     \node(nr2)at(7.2,2)[color=green,place,tokens=1,label=right:$\lnot r2$]{};
     \node(r2)at(8.4,2)[color=green,place,tokens=0,label=right:$r2$]{};
     \node(sn1)at(5.2,3)[color=green,transition] {$n1$};
     \node(sn2)at(5.2,1)[color=green,transition] {$n2$};
     \node(sm1)at(7.8,3)[color=green,transition] {$m1$};
     \node(sm2)at(7.8,1)[color=green,transition] {$m2$};
     \path(g) edge [->] (sn1);
     \path(sn1) edge [->] (ng);
     \path(sn1) edge [->] (r1);
     \path(r1) edge [->] (sn2);
     \path(sn2) edge [->] (nr1);
     \path(nr1) edge [->] (sn1);
     \path [->] (sn2) edge [bend right,out=-40,in=-180,min distance=1cm] (g);
     \path(g) edge [->] (sm1);
     \path(sm1) edge [->] (ng);
     \path(sm1) edge [->] (r2);
     \path(r2) edge [->] (sm2);
     \path(sm2) edge [->] (nr2);
     \path(nr2) edge [->] (sm1);
     \path [->] (sm2) edge [bend left,out=40,in=180,min distance=1cm] (g);
     \path [<-] (sm2) edge [bend left,out=60,in=170,min distance=0.1cm] (ng);
     \path [<-] (sn2) edge [bend left,out=-60,in=-170,min distance=0.1cm] (ng);

     \path [->,>=read] (r1) edge (n1); 
     \path [->,>=read] (r1) edge (n2); 
     \path [->,>=read] (w1) edge (sn1);
     \path [->,>=read] (a1) edge [bend right] (sn2);

     \node(w2)at(11,4)[color=blue,place,tokens=1,label=above:$w2$]{};
     \node(nw2)at(11.3,3.3)[color=blue,place,tokens=0,label=below:$\lnot w2$]{};
     \node(a2)at(12,1)[color=blue,place,tokens=0,label=right:$a2$]{};
     \node(na2)at(11.5,1.5)[color=blue,place,tokens=1,label=right:$\lnot a2$]{};
     \node(t2)at(11,2)[color=blue,place,tokens=0,label=left:$t2$]{};
     \node(nt2)at(9.4,2)[color=blue,place,tokens=1,label=right:$\lnot t2$]{};
     \node(m1)at(10.2,3)[color=blue,transition] {$m1$};
     \node(m2)at(10.2,1)[color=blue,transition] {$m2$};
     \node(m3)at(12,4)[color=blue,transition] {$m3$};
     \path(m3) edge [->] (w2);
     \path(w2) edge [->] (m1);
     \path(m1) edge [->] (t2);
     \path(t2) edge [->] (m2);
     \path(m2) edge [->] (a2);
     \path(a2) edge [->] (m3);
     \path(m3) edge [<-] (nw2);
     \path(nw2) edge [<-] (m1);
     \path(m1) edge [<-] (nt2);
     \path(nt2) edge [<-] (m2);
     \path(m2) edge [<-] (na2);
     \path(na2) edge [<-] (m3);

     \path [->,>=read] (r2) edge (m1); 
     \path [->,>=read] (r2) edge (m2); 
     \path [->,>=read] (w2) edge (sm1);
     \path [->,>=read] (a2) edge [bend left] (sm2);
     \end{tikzpicture}
    \caption{Composition of Petri nets on places.}
    \label{fig:comp_places}
\end{figure}
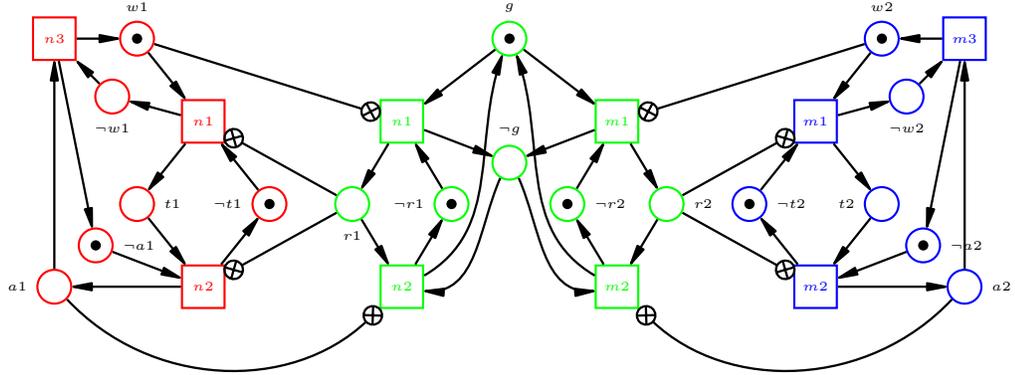

\begin{example}
Fig.~\ref{fig:comp_places} represents the synchronizations as read
arcs. In this model, both $n_1$ and $n_2$ of train $1$ can occur only if $r_1$
is marked, and $n_2$ of the controller can occur only if $a_1$
of train $1$ is marked. In addition, the controller can allow a
train to enter only if the train is waiting for it, therefore the
transition $n_1$ of the controller can occur only if $w_1$ is marked.
Symmetrically for train $2$.
\end{example}
\section{Comparison between the two semantics}
\label{sec:2sem-pn}
In this section, we discuss how to modify a system synchronized on transitions to obtain an equivalent system with synchronizations on data and vice versa (the formal definition of the considered equivalences are presented later in this section).
System synchronized on transitions are generally smaller than those synchronized on data and require to take fewer design choices (e.g. in a shared transition, we do not need to determine which agent is responsible for initiating it). Although this may be more convenient when we need to work on the system at an abstract level, it can create problems if we are interested in the system implementation, where design choices need to be made (as pointed out, for example, in \cite{JPS21}).
On the other hand, the semantics with the synchronization on the data may be more useful when focusing on lower-level features of the system, as it is a closer description of the implementation needs.
Being able to switch between these two semantics allows us to work with the most suitable representation of the system according to our current goal.

In Section~\ref{ssec:syntran2synplac} we start with a system synchronized on transitions and we transform it into a system with data synchronizations. This transformation is the most complex, since we need to transform an abstract system into a more detailed one.
In Theorem \ref{theor:weak-bis}, we prove a weak bisimulation result, showing that we can mimic every behavior of the system synchronized on transitions on its transformation. Theorem \ref{theor:eq_rmark} complete the result proving the equivalence between the reachable markings in the two semantics. 

Section~\ref{ssec:synplac2synact} presents the transformation in the opposite direction, from a system synchronized on data to the one synchronized on transitions. In this case the transformation is more straightforward and allows us to prove the isomorphism between the reachability graphs of the models in two semantics (Theorem \ref{theor:comp-pl-tr}).

\subsection{From synchronization on transitions to synchronization on places}
\label{ssec:syntran2synplac}

In this Section we assume that we have a multi-agent Petri net system constructed as in Sec.~\ref{sec:tr} (with pairwise disjoint sets of transitions), and we show how to transform it into a multi-agent system as defined in Sec.~\ref{sec:pn-pl}.
We provide a transformation at the level of agents.
Namely, for a set of 1-safe Petri nets $\{\Sigma_i \;|\; i\in\Ag\}$ we
construct a set of 1-safe Petri nets with read arcs $\{\Sigma^{init_i}_{int}\}$.

We can do it by applying the following procedure:\begin{enumerate}
    \item We establish the orders between the agents (in more general setting one might want to establish the order between the agents for each shared transition separately). 
    \item We add the synchronization places which guaranties that potential concurrency inside agents is the possibility of equivalent interleavings (not true concurrency). Namely, whenever an agent decides to synchronize with another agent, it is too busy to perform any other action before the synchronization finishes.
    \item For each group of transitions synchronizing on a
    given label $l$, we split each transition labeled with $l$
    into a proper number of copies (one for each agent participating in the synchronized transition in the global system). 
    Each of these copies is then additionally split into two transitions, 
    and add a pair of complementary places in between
    (constructing a binary internal variable).
    \item We add binary external variables related with the groups of synchronized transitions and proper read arcs.
    \item In order for the action to occur, the first agent needs to start it and execute the first of the split transitions.
    \item When the first agent is in the additional intermediate place, only the first half of the same action in the other agents may occur (we need to be sure that everybody follows, before proceeding). Until all the agents executed 
    this half, also every action in the first agent is blocked (by the emptiness of the synchronization place of this agent).
    \item When all the agents did the first half of the action,
    they can execute the second half. 
\end{enumerate}

To simplify further considerations, we assume that all labels present in the simulated system are singletons, and we treat them as single elements, not sets and we use $\varepsilon$ instead of $\emptyset$ for empty labels.

More formally, let $\Sigma_1 = (P_1, T_1, F_1, \iota_1, \ell_1), ..., \Sigma_n = (P_n, T_n, F_n, \iota_n, \ell_n)$ be the set of Petri net agents, and $\Sigma = (P, T, F, m_0, \ell)$ be their common transitions based composition. We will construct a system $\Sigma' = (P', T', F', \read', m'_0, \ell')$ which simulates $\Sigma$ using the data driven composition.

Let us introduce some necessary notions.
By $\Evt^s_i$ we denote the set of all shared events of $i$-th component. Namely, $\Evt^s_i = \{a\in\Evt_i \;|\; a\in\bigcup_{j\neq i} \Evt_j\}$.
By $Rg:\Evt\to VecTR$ we denote a range function which assigns to each (shared) event the set of all combinations of transitions labeled with $a$ that can be chosen from agents.
Formally each $ST\in VecTR$ is a function from the set of agents $\Ag$ to $\{\varepsilon\}\cup\bigcup_{i\in\Ag}\;T_i$ and $ST(i)\in\ell_i^{-1}(a)$ if $a\in\Evt_i$ 
and $ST(i)=\varepsilon$ otherwise.
Moreover, $\mathsf{next}:VecTR\times\Ag\to\Ag$ is the function that for $ST\in VecTR$ and $i\in\Ag$ returns the smallest agent $j>i$ such
that $ST(j)\neq\varepsilon$ if such exists or the smallest $j>0$ such that $ST(j)\neq\varepsilon$ otherwise. 
Note that $\next(ST,|\Ag|)$ returns the smallest (number of) agent involved in $ST$.
Similarly we define $\mathsf{prev}:VecTR\times\Ag\to\Ag$ to return the largest agent with image different from $\varepsilon$ and smaller than specified.
We also denote by $T_i^{unique}$ the set of all transitions of $i$-th agent which have labels occurring only in this agent (namely
$T_i^{unique}=T_i\setminus\ell^{-1}(\Evt^s_i)$) and by $T_i^{shared}$ all other transitions of $i$-th agent.
Then the constructions is defined as follows:
\begin{itemize}
    \item $P^{init_i}_{int} = P_i
    \cup
    \{p_i^{sync}\}
    \cup
    \{t^{(a,ST(j))}_{in}, t^{(a,ST(j))}_{out} \;|\; a\in\Evt^s_i \wedge ST\in Rg(a) \wedge t=ST(i) \wedge ST(j)\neq\varepsilon\}$;
    \item $T^{init_i}_{int} =
    T_i^{unique}
    \cup
    \{t^{(a,ST(j))}_{b},t^{(a,ST(j))}_{e} \;|\; a\in\Evt^s_i \wedge ST\in Rg(a) \wedge t=ST(i) \wedge ST(j)\neq\varepsilon\}
    $;
    \item $\flow^{init_i}_{int} =
    \flow_i|_{(P_i\times T_i^{unique})\cup(T_i^{unique}\times P_i)}
    \cup
    \{(p_i^{sync},t^{(a,ST(i))}_{b}) \;|\; t\in T_i^{shared}\}    
    \cup
    \{(t^{(a,ST(i))}_{e},p_i^{sync}) \;|\; t\in T_i^{shared}\}    
    \cup
    \{(p,t^{(a,ST(i))}_{b}) \;|\; (p,t)\in F_i \wedge t\in T_i^{shared}\}
    \cup
    \{(t^{(a,ST(i))}_{e},p) \;|\; (t,p)\in F_i \wedge t\in T_i^{shared}\}
    \cup
    \{(t^{(a,ST(j))}_{in},t^{(a,ST(j))}_{b}),(t^{(a,ST(j))}_{e},t^{(a,ST(j))}_{in}) \;|\; a\in\Evt^s_i \wedge ST\in Rg(a) \wedge t=ST(i) \wedge ST(j)\neq\varepsilon\}
    \cup
    \{(t^{(a,ST)}_{out},t^{(a,ST)}_{e}),(t^{(a,ST)}_{b},t^{(a,ST)}_{out}) \;|\; t\in T_i^{shared}\}
    $;
    \item $\read^{init_i}_{int} = 
    \{(t'^{(a,ST(j))}_{out},t^{(a,ST(i))}_{b}) \;|\; a\in\Evt^s_i \wedge ST\in Rg(a) \wedge t=ST(i) \wedge t'=ST(j) \wedge j=\mathsf{next}(ST,i))\neq ST(\mathsf{prev}(ST,1))\}
    \cup
    \{(t'^{(a,ST(j))}_{in},t^{(a,ST(i))}_{e}) \;|\; a\in\Evt^s_i \wedge ST\in Rg(a) \wedge t=ST(i) \wedge t'=ST(j) \wedge j=\mathsf{next}(ST,i))\neq ST(\mathsf{prev}(ST,1))\}
    \cup
    \{(t'^{(a,ST(j))}_{in},t^{(a,ST(i))}_{b}) \;|\; a\in\Evt^s_i \wedge ST\in Rg(a) \wedge t=ST(i) \wedge t'=ST(j) \wedge j=\mathsf{prev}(ST,i)) \wedge i=ST(\mathsf{next}(ST,|A|))\}
    \cup
    \{(t'^{(a,ST(j))}_{out},t^{(a,ST(i))}_{e}) \;|\; a\in\Evt^s_i \wedge ST\in Rg(a) \wedge t=ST(i) \wedge t'=ST(j) \wedge j=\mathsf{prev}(ST,i)) \wedge i=ST(\mathsf{next}(ST,|A|))\}
    $;
    \item $\iota^{init_i}_{int}(p) = \iota_i(p)$ for $p\in P_i$
    and
    $\iota^{init_i}_{int}(p) = 1$ for $p\in
    \{t^{(a,ST(j))}_{in} \;|\; a\in\Evt^s_i \wedge ST\in Rg(a) \wedge t=ST(i) \wedge ST(j)\neq\varepsilon\}$
    and
    $\iota^{init_i}_{int}(p_i^{sync})=1$
    and 
    $\iota^{init_i}_{int}(p) = 0$ for $p\in
    \{t^{(a,ST(j))}_{out} \;|\;  a\in\Evt^s_i \wedge ST\in Rg(a) \wedge t=ST(i) \wedge ST(j)\neq\varepsilon\}$;
    \item $\ell^{init_i}_{int}(t) = \ell_i(t)$ for $t\in T_i^{unique}$
    and
    $\ell^{init_i}_{int}(t^{(a,ST(i))}_{b}) = \ell^{init_i}_{int}(t^{(a,ST(i))}_{e}) = \ell_i(t)$ for $t\in T_i^{shared}$
    and
    $\ell^{init_i}_{int}(t^{(a,ST(j\neq i))}_{b}) = \ell^{init_i}_{int}(t^{(a,ST(j\neq i))}_{e}) = \varepsilon$.
\end{itemize}

At this point we are ready to compose the prepared modules using the fusion of sequential components. We treat all pairs $(t_{in}^{(a,ST)},t_{out}^{(a,ST)})$ where $t$ is a shared transition of one of the components (module or larger system) as binary internal variables and all other pairs $(t_{in}^{(a,ST)},t_{out}^{(a,ST)})$ as binary external variables. Note that we have only binary variables, so the bijections present in the compositions using sequential components fusion are unambiguously defined by the initial marking while the order of the compositions itself does not matter. 

Let us also relate the markings of nets $\Sigma$ and $\Sigma'$ by $sp:2^P\to 2^{P'}$ as follows:
\begin{itemize}
    \item 
        $sp(M)(p) = M(p)$ for $p\in P\cap P'$;
    \item
        $sp(M)(p) = 1$ for $p \in \{t^{(a,ST)}_{in} \;|\; a\in\Evt^s_i \wedge ST\in Rg(a) \wedge (t=ST(i) \vee t=ST(\mathsf{prev}(ST,i)))\}$;
    \item 
        $sp(M)(p) = 0$ for $p \in \{t^{(a,ST)}_{out} \;|\; a\in\Evt^s_i \wedge ST\in Rg(a) \wedge (t=ST(i) \vee t=ST(\mathsf{next}(ST,i)))\}$.
\end{itemize}

Finally, we can also relate the computations of nets $\Sigma$ and $\Sigma'$ by defining two morphisms, $\kappa:T\to T'^+$ and $\pi:T'\to 2^\Lambda$, as follows:
\begin{itemize}
    \item $\kappa(t) = t$ if $t\in T_i^{unique}$ for some agent $i$;
    \item $\kappa(t=(t^{i_1},t^{i_2}\ldots t^{i_k})) = (t^{i_1})^{(a,ST(i_1))}_b(t^{i_2})^{(a,ST(i_2))}_b\ldots (t^{i_k})^{(a,ST(i_k))}_b(t^{i_1})^{(a,ST(i_1))}_e$ $(t^{i_2})^{(a,ST(i_2))}_e\ldots (t^{i_k})^{(a,ST(i_k))}_e$ otherwise, for $a=\ell(t)$ and $SP\sim(t^{i_1},t^{i_2}\ldots t^{i_k})$;
    \item $\pi(t) = \ell_i(t)$ if $t\in T_i^{unique}$ for some agent $i$;
    \item $\pi((t^i)^{(a,ST(i))}_e) = \ell_i(t^i)$ for $t^i\in T_i^{shared}$ 
      and $t^i=ST(i)$ for some $ST\in VecTR$ 
      and $i=\mathsf{next}(ST,|A|)$;
    \item $\pi(t) = \emptyset$ otherwise.
\end{itemize}
In words, $\kappa : T \rightarrow T'^+$ and $\pi : T' \rightarrow 2^\Lambda$ are two auxiliary functions that we use to show the relations between $\Sigma$ and $\Sigma'$. In particular, $\kappa$ is a function associating to each transition $t \in T$ in $\Sigma$ the sequence of transitions in $T'$ derived from $t$ in the construction procedure, with $T'$ set of transitions in $\Sigma'$. 
In particular, if $t \in T$ is not a synchronization transition in $\Sigma$, then $\kappa(t) = t$;
otherwise let $i_1, ..., i_n$ be the agents synchronizing on $t$ in $\Sigma$ and let $i_1 < ... < i_n$ be the assigned priority order. By construction, in $\Sigma'$ we have $n \times 2$ copies of $t$, namely for each agent $i_j$, $j\in \{1,...,n\}$, we have the beginning transition $t^{i_j}_b$ and the ending transition $t^{i_j}e$. We define $\kappa(t)$ as the sequence of copies of transition $t$ as they are forced to occur in $\Sigma'$, namely $\kappa(t)= t^{i_1}_b t^{i_n}_b, ..., t^{i_1}_b t^{i_n}_e$.
Function $\pi$ assigns labels to transitions in $T'$ in order to facilitate the analysis of the relation between the labels of $T$ and of $T'$. 
In particular, if $t$ is not a synchronization transition in $T$, then there is a single occurrence of it in $\Sigma'$ and we define $\pi(t) = \ell(t)$. Otherwise, when multiple agents need to synchronize, we consider the synchronization solved when the agent initializing it executes the second copy of its transition (namely $t_e^{i_1}$).
Indeed, if we observe the occurrence of $t^{i_1}_e$, by construction we are sure that all the agents involved in the synchronization joined, and we know that there is no possible way to prevent the agents to also conclude their action.
Hence, we define $\pi(t^{i_1}_e) = \ell(t)$, and $\pi(t) = \emptyset$ in all other cases.

\begin{example}
Consider the system in Fig. \ref{fig:tr2glpna}, where the agents have been synchronized on transition. To transform it in an equivalent system synchronized on places we first decide an order between the agents for every common action, and then we make a copy of every transition involved in the synchronization, to denote its beginning and its end. In Fig. \ref{fig:sync_tr2pl}, the controller (green component) need to initialize all the shared actions, both with read and with blue train. An example of shared transition is $n1$: the controller can start it by executing $n1b$; afterward, it needs to wait for the red train to also start the action, by executing its copy of $n1b$. 
Once all the agents participating in $n1$ have started it, the controller is allowed to end the action ($n1e$), and the synchronization ends.

When transforming a system synchronized on transitions into a system synchronized on data, we have no information on which agent should be responsible to initialize the actions, therefore we can decide it arbitrarily. 
This is a weakness of modeling a system by using synchronization on transitions, since it may lead to unwanted system behaviors \cite{JPS21}.
In our example, the controller can decide which train is allowed to enter the tunnel without verifying that the train is actually ready to proceed.
When the system is designed by using synchronizations on data, this problem is avoided, since priority between agents is explicitly determined.
\end{example}

\begin{figure}[ht]
     \centering
         \begin{tikzpicture}[node distance=1.3cm,>=arrow30,line  width=0.3mm,scale=1.1,auto,bend angle=45,font=\tiny]
     \tikzstyle{place}=[draw,circle,thick,minimum size=4.5mm]
     \tikzstyle{transition}=
                [draw,regular polygon,thick,
                 regular polygon sides=4,minimum size=8mm, inner sep = -2pt]

     \node(w1)at(1.3,5)[color=red,place,tokens=1,label=above:$w1$]{};
     \node(a1)at(.3,.5)[color=red,place,tokens=0,label=left:$a1$]{};
     \node(t1)at(1.3,2)[color=red,place,tokens=0,label=right:$t1$]{};
     \node(n1b)at(2.1,4.5)[color=red,transition] {$n1b$};
     \node(n1m)at(1.7,3.5)[color=red,place,tokens=0,label=right:$ $]{}; 
     \node(n1c)at(2.5,3.5)[color=red,place,tokens=1,label=right:$ $]{}; 
     \node(n1e)at(2.1,2.5)[color=red,transition] {$n1e$};
     \node(pstr1)at(0.7, 2)[color=red,place,tokens=1]{};
     \node(n2b)at(2.1,1.5)[color=red,transition] {$n2b$};
     \node(n2m)at(1.7,.5)[color=red,place,tokens=0,label=right:$ $]{}; 
     \node(n2c)at(2.5,.5)[color=red,place,tokens=1,label=right:$ $]{}; 
     \node(n2e)at(2.1,-.5)[color=red,transition] {$n2e$};
     \node(n3)at(.3,5)[color=red,transition] {$n3$};
     \path(n3) edge [->] (w1);
     \path(w1) edge [->] (n1b);
     \path(n1b) edge [->] (n1m);
     \path(n1c) edge [->] (n1b);
     \path(n1e) edge [->] (n1c);
     \path(n1m) edge [->] (n1e);
     \path(n1e) edge [->] (t1);
     \path(t1) edge [->] (n2b);
     \path(n2b) edge [->] (n2m);
     \path(n2c) edge [->] (n2b);
     \path(n2e) edge [->] (n2c);
     \path(n2m) edge [->] (n2e);
     \path [->] (n2e) edge [bend right,out=30,in=130,min distance=.8cm] (a1);
     \path(a1) edge [->] (n3);
     \path(pstr1) edge [<->] (n3);
     \path [->] (pstr1) edge [bend left, out=10, in=170] (n1b);
     \path [->] (pstr1) edge [bend left, out=-30, in=-150] (n2b);
     \path [->] (n1e) edge [bend left, out=-30, in=-150] (pstr1);
     \path [->] (n2e) edge [bend left, out=10, in=170] (pstr1);

     \node(g)at(5,5)[color=green,place,tokens=1,label=above:$g$]{};
     \node(psc)at(5,2)[color=green,place,tokens=1]{};
     \node(r1)at(4.3,2)[color=green,place,tokens=0,label=left:$r1$]{};
     \node(r2)at(5.7,2)[color=green,place,tokens=0,label=right:$r2$]{};
     \node(sn1b)at(3.5,4.5)[color=green,transition] {$n1b$};
     \node(sn1m)at(3.9,3.5)[color=green,place,tokens=0,label=right:$ $]{}; 
     \node(sn1c)at(3.1,3.5)[color=green,place,tokens=1,label=right:$ $]{}; 
     \node(sn1e)at(3.5,2.5)[color=green,transition] {$n1e$};
     \node(sn2b)at(3.5,1.5)[color=green,transition] {$n2b$};
     \node(sn2m)at(3.9,.5)[color=green,place,tokens=0,label=right:$ $]{}; 
     \node(sn2c)at(3.1,.5)[color=green,place,tokens=1,label=right:$ $]{}; 
     \node(sn2e)at(3.5,-.5)[color=green,transition] {$n2e$};
     \node(sm1b)at(6.5,4.5)[color=green,transition] {$m1b$};
     \node(sm1m)at(6.1,3.5)[color=green,place,tokens=0,label=right:$ $]{}; 
     \node(sm1c)at(6.9,3.5)[color=green,place,tokens=1,label=right:$ $]{}; 
     \node(sm1e)at(6.5,2.5)[color=green,transition] {$m1e$};
     \node(sm2b)at(6.5,1.5)[color=green,transition] {$m2b$};
     \node(sm2m)at(6.1,.5)[color=green,place,tokens=0,label=right:$ $]{}; 
     \node(sm2c)at(6.9,.5)[color=green,place,tokens=1,label=right:$ $]{}; 
     \node(sm2e)at(6.5,-.5)[color=green,transition] {$m2e$};
     \path(g) edge [->] (sn1b);
     \path(sn1b) edge [->] (sn1m);
     \path(sn1c) edge [->] (sn1b);
     \path(sn1e) edge [->] (sn1c);
     \path(sn1m) edge [->] (sn1e);
     \path(sn1e) edge [->] (r1);
     \path(r1) edge [->] (sn2b);
     \path(sn2b) edge [->] (sn2m);
     \path(sn2c) edge [->] (sn2b);
     \path(sn2e) edge [->] (sn2c);
     \path(sn2m) edge [->] (sn2e);
     \path [->] (sn2e) edge [bend right,out=-30,in=175,min distance=2cm] (g);
     \path(g) edge [->] (sm1b);
     \path(sm1b) edge [->] (sm1m);
     \path(sm1c) edge [->] (sm1b);
     \path(sm1e) edge [->] (sm1c);
     \path(sm1m) edge [->] (sm1e);
     \path(sm1e) edge [->] (r2);
     \path(r2) edge [->] (sm2b);
     \path(sm2b) edge [->] (sm2m);
     \path(sm2c) edge [->] (sm2b);
     \path(sm2e) edge [->] (sm2c);
     \path(sm2m) edge [->] (sm2e);
     \path [->] (sm2e) edge [bend left,out=30,in=-175,min distance=2cm] (g);
     \path [->] (sn1e) edge [bend left, out=20, in=160]  (psc);
     \path [->] (sm1e) edge [bend left, out=-20, in=-160] (psc);
     \path (psc) edge [->, bend left, out=-30, in=-150] (sn1b);
     \path (psc) edge [->, bend left, out=30, in=150] (sm1b);
     \path [->] (psc) edge [bend left, out=20, in=160] (sn2b); 
     \path [->] (psc) edge [bend left, out=-20, in=-160] (sm2b);
     \path [->] (sm2e) edge [bend left, out=30, in=150] (psc);
     \path [->] (sn2e) edge [bend left, out=-30, in=-150](psc);
     
     \node(w2)at(8.7,5)[color=blue,place,tokens=1,label=above:$w2$]{};
     \node(a2)at(9.7,.5)[color=blue,place,tokens=0,label=right:$a2$]{};
     \node(t2)at(8.7,2)[color=blue,place,tokens=0,label=left:$t2$]{};
     \node(m1b)at(7.9,4.5)[color=blue,transition] {$m1b$};
     \node(m1m)at(8.3,3.5)[color=blue,place,tokens=0,label=right:$ $]{}; 
     \node(m1c)at(7.5,3.5)[color=blue,place,tokens=1,label=right:$ $]{}; 
     \node(m1e)at(7.9,2.5)[color=blue,transition] {$m1e$};
     \node(pstr3)at(9.3, 2)[color=blue,place,tokens=1]{};
     \node(m2b)at(7.9,1.5)[color=blue,transition] {$m2b$};
     \node(m2m)at(8.3,.5)[color=blue,place,tokens=0,label=right:$ $]{}; 
     \node(m2c)at(7.5,.5)[color=blue,place,tokens=1,label=right:$ $]{}; 
     \node(m2e)at(7.9,-.5)[color=blue,transition] {$m2e$};
     \node(m3)at(9.7,5)[color=blue,transition] {$m3$};
     \path(m3) edge [->] (w2);
     \path(w2) edge [->] (m1b);
     \path(m1b) edge [->] (m1m);
     \path(m1c) edge [->] (m1b);
     \path(m1e) edge [->] (m1c);
     \path(m1m) edge [->] (m1e);
     \path(m1e) edge [->] (t2);
     \path(t2) edge [->] (m2b);
     \path(m2b) edge [->] (m2m);
     \path(m2c) edge [->] (m2b);
     \path(m2e) edge [->] (m2c);
     \path(m2m) edge [->] (m2e);
     \path [->] (m2e) edge [bend left,out=-30,in=-130,min distance=.8cm] (a2);
     \path(a2) edge [->] (m3);
     \path [->] (pstr3) edge [bend left, out=-10, in=-170] (m1b);
     \path [->] (pstr3) edge [bend left, out=30, in=150] (m2b);
     \path [->] (m1e) edge [bend left, out=30, in=150] (pstr3);
     \path [->] (m2e) edge [bend left, out=-10, in=-170] (pstr3);

     \path [->,>=read] (sn1m) edge [bend left,out=10,in=180,min distance=.4cm] (n1b);     
     \path [->,>=read] (sn1c) edge [bend left,out=10,in=140,min distance=.4cm] (n1e);     
     \path [->,>=read] (n1c) edge [bend left,out=10,in=140,min distance=.4cm] (sn1b);     
     \path [->,>=read] (n1m) edge [bend left,out=10,in=180,min distance=.4cm] (sn1e);     
     \path [->,>=read] (sn2m) edge [bend left,out=10,in=180,min distance=.4cm] (n2b);     
     \path [->,>=read] (sn2c) edge [bend left,out=10,in=140,min distance=.4cm] (n2e);     
     \path [->,>=read] (n2c) edge [bend left,out=10,in=140,min distance=.4cm] (sn2b);     
     \path [->,>=read] (n2m) edge [bend left,out=10,in=180,min distance=.4cm] (sn2e);     
     \path [->,>=read] (sm1m) edge [bend right,out=-10,in=-180,min distance=.4cm] (m1b);     
     \path [->,>=read] (sm1c) edge [bend right,out=-10,in=-140,min distance=.4cm] (m1e);     
     \path [->,>=read] (m1c) edge [bend right,out=-10,in=-140,min distance=.4cm] (sm1b);     
     \path [->,>=read] (m1m) edge [bend right,out=-10,in=-180,min distance=.4cm] (sm1e);     
     \path [->,>=read] (sm2m) edge [bend right,out=-10,in=-180,min distance=.4cm] (m2b);     
     \path [->,>=read] (sm2c) edge [bend right,out=-10,in=-140,min distance=.4cm] (m2e);     
     \path [->,>=read] (m2c) edge [bend right,out=-10,in=-140,min distance=.4cm] (sm2b);     
     \path [->,>=read] (m2m) edge [bend right,out=-10,in=-180,min distance=.4cm] (sm2e);     
     
     \end{tikzpicture}
     \caption{System equivalent to the one in Fig. \ref{fig:tr2glpna} in the sense of Theorem \ref{theor:weak-bis}, where synchronizations on transitions has been transformed in synchronization on places.}
     \label{fig:sync_tr2pl}
\end{figure}
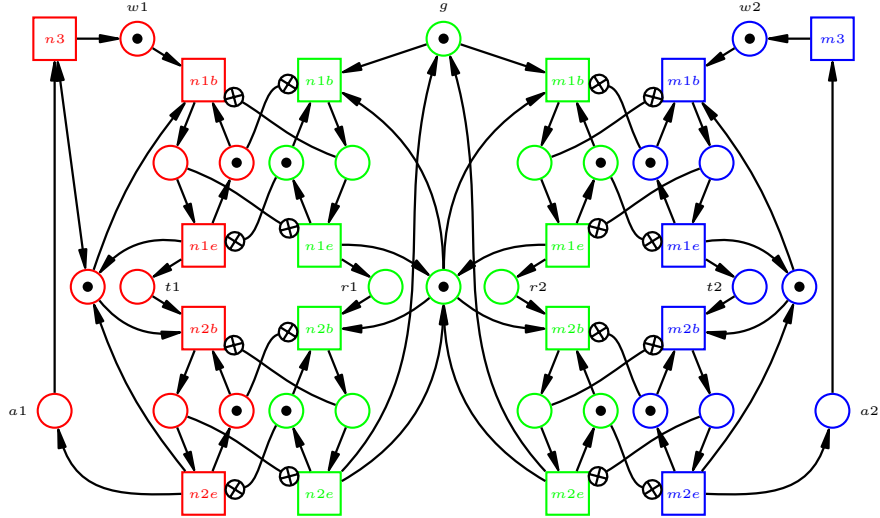

By $Parikh:X^*\to N^X$ we denote the function that assigns to a sequence the Parikh vector (multiset) of the occurrences of the elements.
The provided composition satisfies the weak bisimulation in the following sense.

\begin{theorem}
\label{theor:weak-bis}
Let $\Sigma$ be the global net obtained as the synchronization on transitions of Petri net agents $\Sigma_1,\ldots,\Sigma_n$. Moreover, let $\Sigma'$ be a system which simulates $\Sigma$ using data driven composition as described in Section~\ref{ssec:syntran2synplac}. 
Then for any $M_1,M_2 \in[m_0\rangle$ we have 
$sp(M_1)[v\rangle sp(M_2)$ iff $M_1[u\rangle M_2$
where $Parikh(v) = Parikh(\kappa(u))$ and $\ell'(\pi(v))=\ell(u)$.
\end{theorem}
\begin{proof}
    Let $\Sigma$ be a Petri net obtained from systems $\Sigma_1,\ldots,\Sigma_n$ by the transition driven composition, while $\Sigma'$ be a system obtained by the data driven composition of systems $\Sigma'_1,\ldots,\Sigma'_n$ obtained from $\Sigma_1,\ldots,\Sigma_n$ and $\Sigma$ as described in this Section.

    $(\Leftarrow):$
    Let $M_1[t\rangle M_2$ where $t$ is a shared action associated with $ST$. 
    Then we can easily see that $\kappa(t)$ is enabled at
    $sp(M_1)$ since all input places for shared transition are in all the Petri net agents marked (all nonempty $ST(i)$ is enabled in agent $i$ at $M|_{\Sigma_i}$) and all the $sp(M_1)(t_{in}^{(\ell(t),t=ST(i))})=1$, while for the appropriate complementary places $sp(M_1)(t_{out}^{(\ell(t),t=ST(i))})=0$. Note that executing transition $t^{(\ell(t),t=ST(i))}_b$ changes states of those places to the opposed ones. 
    We only need to check all the read arcs.
    Note that, according to the construction, one of the agents taking part in shared transition is featured (namely the agent $i$ with the smallest number). The transition $t^{(\ell(t),t=ST(i))}_b)$ is connected by a read arc with place $ST(j)_{in}^{(\ell(t),ST(j))}$ (for the agent $j$ with the largest number), which is marked. Transitions $ST(k)^{(\ell(t),ST(k))}_b$ for all the other agents taking part in this shared action a connected by read arcs with places $ST(l)_{out}^{(\ell(t),ST(l))}$ where $l=\mathsf{prev}(ST,k)$. Hence we can, one by one, execute them in order consistent with the order of participating agents, simultaneously consuming appropriate tokens. We proceed with the second second part of the sequence (namely transitions $t^{(\ell(t),t=ST(i))}_e$) similarly, this time putting the tokens to appropriate places. At the end we are indeed at marking $sp(M_2)$.
    
    Note also that $\pi(\kappa(t)) = \ell'((t_{i=\mathsf{next}(ST,|A|)}^{(\ell(t),ST(i))})_{out})$ and 
    $\ell(t) = \ell'((t_{i=\mathsf{next}(ST,|A|)}^{(\ell(t),ST(i))})_{out})$. 
    
    We can proceed this way with all the transitions of $u$ and inductively show
    that if $M_1[u\rangle M_2$ then $sp(M_1)[\kappa(u)\rangle sp(M_2)$ and
    $\ell'(\pi(\kappa(u)))=\ell(u)$.

    $(\Rightarrow):$
    Let $sp(M_1)[v\rangle sp(M_2)$. Let $u=\pi(v)$ and let $(t_i)^{(\ell(t_i),ST(i))}_e = u[1] = v[k]$. If $u[1]$ is a unique transition of agent $i$ then we can make use of the effect caused by place $p_i^{sync}$ and move $v[k]$ to the beginning of a sequence $v$ without affecting the rest of the computation. Let us concentrate on the other case, where $u[1]$ is a shared transition. 
    Since $(t_i)_{out}^{(\ell(t),ST(i))}$ is a preplace of $(t_i)^{(\ell(t_i),ST(i))}_e$ for any $ST$ and 
    $sp(M_1)((t_i)_{out}^{(\ell(t),ST(i))}) = 0$ and the only transition that changes the value of this place is $(t_i)_b^{(\ell(t_i),SP(i))}$, we know that it must have been executed in $v$ before $v[k]$. Moreover, only one instance of this transition may be executed before $v[k]$ as $v[k]$ is the first occurrence of transition with $e$ in subscript (remember that it is the action of the smallest agent for any shared action labeled by $\ell(t)$).
    For similar reasons (taking into account appropriate read arcs), there is precisely one occurrence of transition $(t_j)_b^{(\ell(t_i),SP(j))}$ before $v[k]$, where $j$ is the largest (number of) agent participating in the shared transition $t$.
    Indeed, we have that there are read arcs between places $(t_j)_{out}^{(\ell(t),ST(j))}$ (again for any $ST$) and the only possibility to change their value from initial $0$ to $1$ is to execute transition $(t_j)_b^{(\ell(t_i),SP(j))}$.
    Now we are ready to fill the gap and argue that between the first occurrence of $(t_i)_b^{(\ell(t_i),SP(i))}$ and the first occurrence of $(t_j)_b^{(\ell(t_i),SP(j))}$ we need to have the transitions for all other agents participating in this shared transition.
    Since $v[k]$ is the first transition of $\kappa(t)$ that put a new token to the places of system $\Sigma_i$ (because of $p_i^{sync}$), we can move all the considered so far components of shared transition $t$ to the beginning of the sequence $v$ obtaining $v'$ such that 
    $sp(M_1)[v'\rangle sp(M_2)$.

    Since at the end we obtain the marking $sp(M_2)$, in the sequence $v$ there have to be transitions that will restore the places $(t_k)_{in}^{(\ell(t),ST(k))}$ and $(t_k)_{out}^{(\ell(t),ST(k))}$ (for all possible $ST$ at once). 
    It remains to note that those transitions also can be moved to the front, as they only puts new tokens to the places of system $\Sigma$ (and since we were able to initialize all the components of shared transition $t$, namely all the preconditions were satisfied, 1-safeness of the Petri net agents guarantees that appropriate places are empty). As a result we obtain $v''=\kappa(t)v'''$ which is $v$ with $\kappa(t)$ moved to the front (the same for the case of unique transition described earlier). 
    Hence $sp(M_1)[t\rangle sp(M_3)[v'''\rangle sp(M_2)$ and $M_1[t\rangle M_3$.

    As we can repeat the same reasoning for the sequence $v'''$ (which has image of $\pi$ shorter then $v$) and inductively obtain a sequence $w$ equivalent to $v$ such that $\pi(w) = \pi(v)$ and $Parikh(w) = Parikh(v)$ and there exists $u$ such that $w=\kappa(u)$ and $M_1[u\rangle M_2$.
    Hence we have shown that if $sp(M_1)[v\rangle sp(M_2)$ then there exist
    an appropriate $u$ that $M_1[u\rangle M_2$.
\end{proof}

We can also formulate a fact associating the markings reachable by the original systems and the transformed one:

\begin{theorem}
\label{theor:eq_rmark}
Let $\Sigma$ be the global net obtained as the synchronization on transitions of Petri net agents $\Sigma_1,\ldots,\Sigma_n$. Moreover, let $\Sigma'$ be a system which simulates $\Sigma$ using data driven composition as described in Section~\ref{ssec:syntran2synplac}.
Then $M\in[m_0\rangle$ if and only if $M'=sp(M)\in[m'_0\rangle$.
\end{theorem}
\begin{proof}
    Again, let $\Sigma$ be a Petri net obtained from systems $\Sigma_1,\ldots,\Sigma_n$ by the transition driven composition, while $\Sigma'$ be a system obtained by the data driven composition of systems $\Sigma'_1,\ldots,\Sigma'_n$ obtained from $\Sigma_1,\ldots,\Sigma_n$ and $\Sigma$ as described in this Section.

    We start by checking that $sp(m_0) = m'_0$. By the construction we know that for all $p\in P'_i\cap P_i$ we copied the initial marking, hence $sp(m_0)(p) = m'_0(p)$.
    Moreover, the only new places that 'survived' after the all fusions are those related to the internal variables. Namely $(t_{in}^{(a,ST(i))}$ and $t_{out}^{(a,ST(i))}$ for any $t\in T\setminus(\bigcup_j T^{unique})$. By the construction, we have $m'_0(t_{in}^{(a,ST(i))})=1$ and $m'_0(t_{out}^{(a,ST(i))})=0$ for all those places. Hence, by the definition of $sp$, $m'_0(t_{in}^{(a,ST(i))})=sp(m_0)(t_{in}^{(a,ST(i))})$ and $m'_0(t_{out}^{(a,ST(i))})=sp(m_0)(t_{out}^{(a,ST(i))})$.

    To finish the structural induction we need to repeat the reasoning presented in the proof of Theorem~\ref{theor:weak-bis}.
    
\end{proof}

\subsection{From synchronization on places to synchronization on transitions}
\label{ssec:synplac2synact}
In this section we assume to have a multi-agent Petri net system constructed as in Sec. \ref{sec:pn-pl}, and we show how to transform it
into a multi-agent system as defined in Sec. \ref{sec:tr}.
We produce the transformation at the level of the agents, namely, for each agent $\Sigma_{i,int}$,
derived from a module as in Def. \ref{d:modulef},
we construct a $1$-safe agent $\Sigma_i$;
once that the set of agents $\Sigma_1$,..., $\Sigma_n$ has been constructed, we apply the 
construction of the global model defined in Sec. \ref{sec:tr}, with synchronization of common actions.
Theorem \ref{theor:comp-pl-tr} proves the soundness of our transformation.

Let $\Sigma_{i, int}$ be an agent derived from a module.
To get the new agent $\Sigma_i$ we first start from the the subsystem $\Sigma_{i,var}$ of $\Sigma_{i,int}$
as defined in Sec. \ref{sec:pn-pl}.
Then, for each internal variable $x$ of $\Sigma_{i, var}$, we check which agents use $x$ as external variable.
Let $S =\{\Sigma_{j, int},...,\Sigma_{m,int}\}$ this group of agents.
For each agent $\Sigma_{l,int} \in S$, let $P_{ext,x,l}$ be the set of places spanning the value of variable $x$ in $\Sigma_{l,int}$. 
We need to check which transitions in $\Sigma_{l,var}$ are connected with read arcs to places in $P_{ext, x, l}$.
For each transition $t\in \Sigma_{l,var}$ connected with a place $p\in P_{ext,x, l}$ we
add a transition $t'$ to $\Sigma_i$ such that
$\ell(t) = \ell(t')$, and $\pre{t'} = \post{t'} = \{p\}$. This construction is well-defined since by construction a copy of $p$ must be present in $\Sigma_{i,var}$.
We repeat the same procedure for all the agents in the system.
At the end we synchronize the agents $\Sigma_1$,...,$\Sigma_n$ as described in Sec. \ref{sec:tr}.
\begin{example}
Fig. \ref{fig:data2tr_sync} shows an example of transformation from a system synchronized on data (in Fig. \ref{fig:comp_places}) to a system synchronized on transitions. For every transition needing an external variable to occur, we add to the agent whose variable belong a new transition with the same name, connected to the place with a self loop.
For example, the train can enter the tunnel (transition $n1$) only when it has been allowed by the controller (which need to be in place $r1$). Hence, we add to the controller transition $n1$, and connect it to place $r1$ with a self loop. In this way, inside the controller agent, transition $n1$ can occur only when $r1$ is true.
This condition remains true when the transition $n1$ belonging to the controller is fused with transition $n1$ belonging to the red training, guaranteeing the desired behavior.
\end{example}
\begin{figure}[ht]
     \centering
         \begin{tikzpicture}[node distance=1.3cm,>=arrow30,line  width=0.3mm,scale=1.1,auto,bend angle=45,font=\tiny]
     \tikzstyle{place}=[draw,circle,thick,minimum size=4.5mm]
     \tikzstyle{transition}=
                [draw,regular polygon,thick,
                 regular polygon sides=4,minimum size=8mm, inner sep = -2pt]

     \node(w1)at(2,4)[color=red,place,tokens=1,label=above:$t : w1$]{};
     \node(nw1)at(1.7,3.3)[color=red,place,tokens=0,label=below:$t : \lnot w1$]{};
     \node(a1)at(1,1)[color=red,place,tokens=0,label=left:$t : a1$]{};
     \node(na1)at(1.5,1.6)[color=red,place,tokens=1,label=below:$t : \lnot a1$]{};
     \node(t1)at(2,2)[color=red,place,tokens=0,label=above:$t : t1$]{};
     \node(nt1)at(3.6,2)[color=red,place,tokens=1,label=left:$t : \lnot t1$]{};
     \node(n1)at(2.8,3)[color=red,transition] {$t : n1$};
    \node(n1g)at(3,3)[color=green,transition] {$s : n1$};
     \node(n2)at(2.8,1)[color=red,transition] {$t : n2$};
     \node(n2g)at(3,1)[color=green,transition] {$s : n2$};
     \node(n3)at(1,4)[color=red,transition] {$t : n3$};
     \path(n3) edge [->] (w1);
     \path(nw1) edge [->] (n3);
     \path(w1) edge [->] (n1);
     \path(n1) edge [->] (nw1);
     \path(n1) edge [->] (t1);
     \path(t1) edge [->] (n2);
     \path(n2) edge [->] (a1);
     \path(na1) edge [->] (n2);
     \path(a1) edge [->] (n3);
     \path(n3) edge [->] (na1);
     \path(nt1) edge [->] (n1);
     \path(n2) edge [->] (nt1);

     \node(g)at(6.5,4)[color=green,place,tokens=1,label=above:$s : g$]{};
     \node(ng)at(6.5,2.5)[color=green,place,tokens=0,label=above:$s : \lnot g$]{};
     \node(r1)at(4.6,2)[color=green,place,tokens=0,label=below:$s : r1$]{};
     \node(nr1)at(5.8,2)[color=green,place,tokens=1,label=left:$s : \lnot r1$]{};
     \node(nr2)at(7.2,2)[color=green,place,tokens=1,label=right:$s : \lnot r2$]{};
     \node(r2)at(8.4,2)[color=green,place,tokens=0,label=below:$s : r2$]{};
     \node(sn1)at(5.2,3)[color=green,transition] {$s:n1$};
     \node(sn1r)at(5,3)[color=red,transition] {$t:n1$};
     \node(sn2)at(5.2,1)[color=green,transition] {$s:n2$};
     \node(sn2r)at(5,1)[color=red,transition] {$t:n2$};
     \node(sm1)at(7.8,3)[color=green,transition] {$s:m1$};
     \node(sm2)at(7.8,1)[color=green,transition] {$s:m2$};
     \node(sm1b)at(8,3)[color=blue,transition] {$t:m1$};
     \node(sm2b)at(8,1)[color=blue,transition] {$t:m2$};
     \path(g) edge [->] (sn1);
     \path(sn1) edge [->] (ng);
     \path(sn1) edge [->] (r1);
     \path(r1) edge [->] (sn2);
     \path(sn2) edge [->] (nr1);
     \path(nr1) edge [->] (sn1);
     \path [->] (sn2) edge [bend right,out=-40,in=-180,min distance=1cm] (g);
     \path(g) edge [->] (sm1);
     \path(sm1) edge [->] (ng);
     \path(sm1) edge [->] (r2);
     \path(r2) edge [->] (sm2);
     \path(sm2) edge [->] (nr2);
     \path(nr2) edge [->] (sm1);
     \path [->] (sm2) edge [bend left,out=40,in=180,min distance=1cm] (g);
     \path [<-] (sm2) edge [bend left,out=60,in=170,min distance=0.1cm] (ng);
     \path [<-] (sn2) edge [bend left,out=-60,in=-170,min distance=0.1cm] (ng);

     \path [<->] (r1) edge (n1g); 
     \path [<->] (r1) edge (n2g); 
     \path [<->] (w1) edge (sn1r);
     \path [<->] (a1) edge [bend right] (sn2r);

     \node(w2)at(11,4)[color=blue,place,tokens=1,label=above:$t:w2$]{};
     \node(nw2)at(11.3,3.3)[color=blue,place,tokens=0,label=below:$t:\lnot w2$]{};
     \node(a2)at(12,1)[color=blue,place,tokens=0,label=right:$t:a2$]{};
     \node(na2)at(11.5,1.6)[color=blue,place,tokens=1,label=below:$t:\lnot a2$]{};
     \node(t2)at(11,2)[color=blue,place,tokens=0,label=above:$t:t2$]{};
     \node(nt2)at(9.4,2)[color=blue,place,tokens=1,label=right:$t:\lnot t2$]{};
     \node(m1)at(10.2,3)[color=blue,transition] {$t:m1$};
     \node(m2)at(10.2,1)[color=blue,transition] {$t:m2$};
     \node(m1g)at(10,3)[color=green,transition] {$s:m1$};
     \node(m2g)at(10,1)[color=green,transition] {$s:m2$};
     \node(m3)at(12,4)[color=blue,transition] {$t:m3$};
     \path(m3) edge [->] (w2);
     \path(w2) edge [->] (m1);
     \path(m1) edge [->] (t2);
     \path(t2) edge [->] (m2);
     \path(m2) edge [->] (a2);
     \path(a2) edge [->] (m3);
     \path(m3) edge [<-] (nw2);
     \path(nw2) edge [<-] (m1);
     \path(m1) edge [<-] (nt2);
     \path(nt2) edge [<-] (m2);
     \path(m2) edge [<-] (na2);
     \path(na2) edge [<-] (m3);

     \path [<->] (r2) edge (m1g); 
     \path [<->] (r2) edge (m2g); 
     \path [<->] (w2) edge (sm1b);
     \path [<->] (a2) edge [bend left] (sm2b);
     \end{tikzpicture}
     \caption{Transformation of the system to obtain a global model with synchronizations on transitions equivalent (in the sense of Theorem \ref{theor:comp-pl-tr}) to the system obtained by synchronizing sequential components in Fig. \ref{fig:comp_places}.}
     \label{fig:data2tr_sync}
\end{figure}
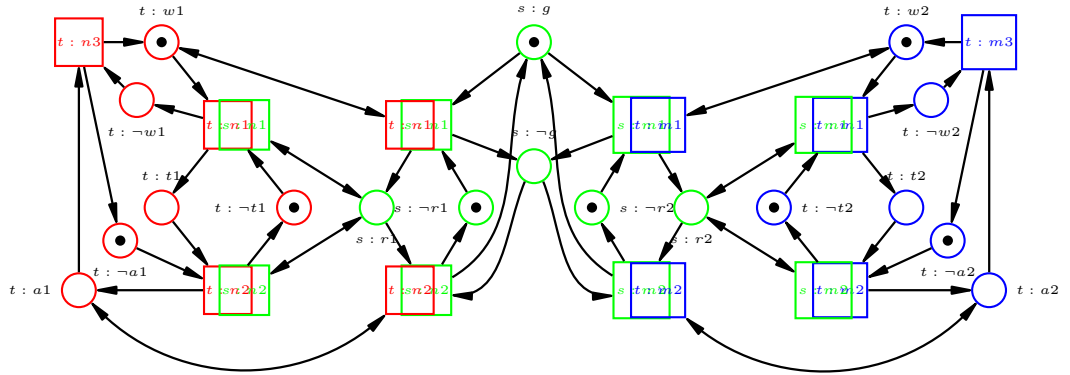
\begin{theorem}
\label{theor:comp-pl-tr}
Let $\Sigma$ be the global net obtained as the synchronization on the common sequential components of Petri net agents $\Sigma_{1,int},\ldots,\Sigma_{n,int}$. Moreover, let $\Sigma_i$, $i\in\{1,...,n\}$, be Petri net agents obtained from $\Sigma_{i,int}$ as described in Section~\ref{ssec:synplac2synact} and
$\Sigma'$ be the global Petri net obtained from $\Sigma_1,\ldots,\Sigma_n$ by synchronization on transitions. Then \[MG(\Sigma) \equiv MG(\Sigma').\]
\end{theorem}
\begin{proof}
Let $\Sigma = (P,T,F,m_0,\ell)$ and $\Sigma' = (P,T',F',m'_0,\ell)$.

First note that, by construction, we have that for all $\Sigma_i$, the initial marking is the same to the one of $\Sigma_{i,var}$. Hence, $M_0 = M'_0$. 
It remains to prove that $M[t\rangle_\Sigma M'$ iff $M[t\rangle_{\Sigma'}M'$.

Let $t\in T$ be a transition of Petri net agent $\Sigma_{i,int}$ and $M\in 2^P$ (which is equivalent formulation of $M:P\to\{0,1\}$). Since $t$ is enabled at $M$, we have that $t$ is enabled at $M|_{\Sigma_{i,int}}$. This means that all $p\in P$ such that $(p,t)\in F$ are marked. Obviously, this means that $t$ is enabled at $M|_{\Sigma_i}$.
It remains to examine, whether for any $\Sigma_j\neq\Sigma_i$ such that $t'\in T_j$ and $\ell(t)=\ell(t')$ we have $t'$ enabled at $M|_{\Sigma_j}$. By the construction, we have $(p,t)\in F_j$ for some $p\in P$. If $t$ is enabled at $M|_{\Sigma_{i,int}}$ and
$t'\in T_{j,int}$ then, by the composition, $(p,t')\in \read_{j,int}$ and there exists a place $p'\in P_{i,int}$ associated by appropriate synchronization on sequential components with $p\in P_{j,var}$, and both $p$ and $p'$ are marked in $M_{\Sigma_{i,int}}$ and $M_{\Sigma_{j,int}}$, appropriately. Hence, by the construction, we have $p'$ marked at $M_{\Sigma_j}$. Hence $t'$ is enabled in $\Sigma_j$ at $M_{\Sigma_j}$.

Let us now assume that $t$ is enabled in $\Sigma$ at $M$.
This means that $t\in T'$ is a transition of Petri net agent $\Sigma_i$. Then we have two cases: (1) $t\in T_{i,int}$ or (2) $t\in T_{j,int}$ for $j\neq i$ and there exists $p\in P_{i,int}$ associated with $p'\in P_{j,int}$ such that $(p',t)\in\read_{j,int}$. However, in the second case we have $t'\in T_j$ with $\ell(t') = \ell(t)$ and we can limit ourselves to the first case.
Note that, directly by the construction, $t$ is enabled in $\Sigma_{i,var}$ at $M|_{\Sigma_{i,var}}$. We only need to argue that if $(p,t)\in\read_{i,int}$ then $p$ is marked at $M|_{\Sigma_{i,int}}$. Indeed, place $p$ need to be synchronized on sequential component with appropriate $p'\in P_{j,var}$. Hence there exists $t'\in T_j$ such that $\ell(t') = \ell(t)$ and $(p',t')\in F_j$. By the construction, $t'$ is synchronized with $t$, hence also enabled in $\Sigma_j$ at $M|_{\Sigma_j}$. As a result, $p'$ is marked at $M|_{\Sigma_j}$ and $M|_{\Sigma_{j,var}}$. Hence $p$ is marked in $\Sigma_{i,int}$ at $M|_{\Sigma_{i,int}}$.
Hence $t$ is enabled in $\Sigma'$ at $M$.

The values of the resulting markings are simple consequences of construction and compositions.
\end{proof}

\section{Conclusion}
\label{sec:concl}
In this paper we propose a framework based on 1-safe Petri nets to model asynchronous multi-agent systems.
In particular, we define two semantics, considering agent synchronizations on transitions and on data. 
Both these semantics were previously defined and studied for AMAS modeled as transition systems, but  the relation between the two semantics has been scarcely investigated in the literature.
For each of the two semantics, we show that we can define a composition operation on Petri net agents such that the marking graph of the resulting global system is isomorphic to the underlying graph obtained by composing the agents modeled as transition systems.
In the case of synchronization on transitions, the synchronization operation is straightforward. However, if we consider synchronization on data, we propose a new version of fusion for Petri nets, namely fusion of sequential components which was not, up to our best knowledge, considered before.
This new synchronization can be used in future works to study assume-guarantee reasoning on Petri nets.

The two semantics allow us to model systems with a different level of abstraction, and can be both useful depending on the task that need to be performed on it. In particular, synchronization of transitions provides a simpler and higher level of abstraction model, whereas synchronization on data results in a lower lever global model, which can be helpful for the system implementation.
In this paper, we show how to switch from one semantics to the other, so that it is always possible to use the most suitable semantics for the current task. 
While switching from a system synchronized on data to a system synchronized on transition is always possible and does not cause any particular issue, the reverse process requires to explicitly make design decisions that were not needed on the system synchronized on transitions, and this can introduce potential unwanted behaviors.
An example of this is the decision on which agent needs to initiate a synchronization. A bad coordination between the decisions of the leading agents can bring to deadlock situations that are hidden in the system synchronized on transitions. To see this, consider two conflicting actions $a1$ and $a2$, both shared by Agent 1 and Agent 2. 
If Agent 1 want to start $a1$ and agent 2 want to start $a2$, each of the two agents could start its main action without joining the execution of the other, causing a deadlock.
In a real implemented system, this aspect needs to be taken into account, therefore, switching the models between the two semantics to check for the presence of unwanted behaviors can be very useful.
In this paper we established a total hierarchy between the agents, but in some cases more flexible solutions may be desirable. In future works we plan to determine the weakest conditions that allow to avoid this kind of deadlocks. 

One can see that the proposed transformations are not expensive.
The size of the model grows linearly in size with both the transformations proposed in Sec. \ref{sec:2sem-pn}: when we transform from a system synchronized on data to a system synchronized on transitions, the dimension of the global net remains the same, and the only modifications are made at the level of the agents; in the opposite direction, for each shared transition we need to produce $2 \times n$ copies, where $n$ is the number of agents involved in the synchronization of that transition, and $n$ new places.
However, iterating the procedure produces larger and larger systems. 
We plan to work on the compositions and the transformations in order to avoid this, and propose operations so that one transformation is the reverse of the other.

Finally, since one of the applications of our work is the model-checking of concurrent systems, we plan to develop methods to improve its efficient. For example, we plan to modify our algorithms to first decompose the Petri net agents into sequential components. In many cases, verifying certain properties in concurrent systems does not require constructing the entire global system, as some of its parts are independent of each other.  Decomposing the agents before the composition may allow us to obtain smaller systems in which the properties can be verified.

\bibliographystyle{plain}
\bibliography{MAS}

\end{document}